%% file: arxiv.tex
\documentclass[a4paper, 11pt]{article}
\usepackage{fullpage}
\usepackage{amsmath, amssymb, amsfonts, amsthm, mathtools}
\usepackage{enumerate}
\usepackage{xfrac}
\usepackage[dvipsnames]{xcolor}
\usepackage{hyperref}
\hypersetup{
    colorlinks=true,
    linkcolor=NavyBlue,
    urlcolor=NavyBlue,
    citecolor=NavyBlue,
    bookmarksopen=false,
    frenchlinks=false,
    pdftitle={The Computational Complexity of the Housing Market},
}

\usepackage[capitalize]{cleveref}
\usepackage[authoryear]{natbib}
\usepackage{booktabs}
\usepackage{tikz, pgfplots}
\pgfplotsset{compat=newest}
\usetikzlibrary{shapes, shapes.misc, backgrounds, math}

\newtheorem{theorem}{Theorem}[section]
\newtheorem{lemma}[theorem]{Lemma}
\newtheorem{proposition}[theorem]{Proposition}

\newtheorem{corollary}[theorem]{Corollary}
\theoremstyle{definition}
\newtheorem{definition}{Definition}[section]

\DeclareMathOperator{\R}{{\mathbb{R}}}
\DeclareMathOperator{\N}{{\mathbb{N}}}
\DeclareMathOperator{\Z}{{\mathbb{Z}}}

\DeclareMathOperator{\conv}{{conv}}
\DeclareMathOperator{\argmin}{{argmin}}
\DeclareMathOperator{\poly}{{poly}}

\usepackage{bm}
\newcommand{\pb}{\bm{p}}
\newcommand{\qb}{\bm{q}}
\newcommand{\xb}{\bm{x}}
\newcommand{\yb}{\bm{y}}
\newcommand{\zb}{\bm{z}}
\newcommand{\vb}{\bm{v}}
\newcommand{\wb}{\bm{w}}
\newcommand{\eb}{\bm{e}}
\renewcommand{\L}{\mathcal{L}}

\newcommand{\eol}{\textsc{EndOfLine}}
\newcommand{\housing}{\textsc{HousingMarket}}
\newcommand{\kkm}{\textsc{KKM}}
\newcommand{\rkkm}{\textsc{RainbowKKM}}
\newcommand{\cake}{\textsc{CakeCutting}}
\newcommand{\spernerX}{\textsc{2D-Sperner}}
\newcommand{\sperner}{\textsc{Sperner}}

\usepackage[most]{tcolorbox}
\usepackage{varwidth}
\newtcolorbox{problem}[2][]{empty, 
coltitle=black,
fonttitle=\bfseries\sffamily,
attach boxed title to top left={yshift=-2.5mm},
boxed title style={empty, size=small, top=1mm, bottom=0pt},
varwidth boxed title=0.3\linewidth,
frame code={
\path (title.east|-frame.north) coordinate (aux);
\path[draw=black, line width=0.3mm, rounded corners]
(frame.west) |- ([xshift=-2.5mm]title.north east) to[out=0, in=180] ([xshift=7.5mm]aux)-|(frame.east)|-(frame.south)-|cycle;  
},
title={#2},#1}

\title{
The Computational Complexity of the Housing Market%
\thanks{This project has received funding from the European Research Council (ERC) under the European Union’s Horizon 2020 research and innovation programme (grant agreement No.~949699). This material is based upon work supported by the National Science Foundation under Grant No. DMS-1928930 and by the Alfred P. Sloan Foundation under grant G-2021-16778, while Teytelboym was in residence at the Simons Laufer Mathematical Sciences Institute (formerly MSRI) in Berkeley, California, during the Fall 2023 semester.
We are grateful to Alex Hollender and Aviad Rubinstein for their comments.
}}

\author{%
Edwin Lock%
\footnote{Department of Economics, University of Oxford, \href{mailto:edwin.lock@economics.ox.ac.uk}{edwin.lock@economics.ox.ac.uk}}
\and
Zephyr Qiu%
\footnote{ETH Zürich, \href{mailto:zephyr.qiu.cn@gmail.com}{zephyr.qiu.cn@gmail.com}}%
\and Alexander Teytelboym%
\footnote{Department of Economics, University of Oxford, \href{mailto:alexander.teytelboym@economics.ox.ac.uk}{alexander.teytelboym@economics.ox.ac.uk}}%
}

\begin{document}

\maketitle

\begin{abstract}
We prove that the classic problem of finding a competitive equilibrium in an exchange economy with indivisible goods, money, and unit-demand agents is PPAD-complete. In this ``housing market'', agents have preferences over the house and amount of money they end up with, but can experience income effects. Our results contrast with the existence of polynomial-time algorithms for related problems: Top Trading Cycles for the ``housing exchange'' problem in which there are no transfers and the Hungarian algorithm for the ``housing assignment'' problem in which agents' utilities are linear in money. Along the way, we prove that the Rainbow-KKM problem, a total search problem based on a generalization by Gale of the Knaster–Kuratowski–Mazurkiewicz lemma, is PPAD-complete. Our reductions also imply bounds on the query complexity of finding competitive equilibrium.
\end{abstract}

\section{Introduction}
The trade of indivisible objects (e.g.,~houses) among unit-demand agents is one of the classic settings in market design. 
If there is no divisible num\'eraire commodity such as money---we refer to such an economy as a ``house exchange''---then Gale's Top Trading Cycles algorithm  finds the unique allocation in the weak core \citep{shapley1974cores,roth1977weak}. 
If agents' utility functions are linear in money---we refer to such a transferable utility economy as a ``housing assignment''---then an efficient allocation can be found by a straightforward solution (e.g.,~the Hungarian algorithm) to an assignment game \citep{koopmans1957assignment,shapley1971assignment}. 
However, both the assumptions of no transfers or of quasilinear utility are somewhat extreme, especially when one considers real-world trade of high-value objects such as houses. 
A more realistic assumption is that agents' utilities are non-linear in money, that is, agents experience income effects and their willingness to pay for a house might depend on their level of wealth (i.e., amount of money or the value of their own house). 
Surprisingly, a competitive equilibrium allocation always exists in such a ``housing market'', a result shown (under various assumptions) by \citet{quinzii1984core}, \citet{gale1984equilibrium} and \citet{svensson1984competitive}.
But while the computational properties of the Top Trading Cycles algorithm in the house exchange model and the Hungarian algorithm in the house assignment model have been known for decades, nothing is known about the computational complexity of finding equilibrium in the housing market.

In this paper, we settle the computational complexity of the housing market in a general version introduced by \citet{gale1984equilibrium}. In particular, we show that \housing{}, the computational problem of finding an approximate competitive equilibrium in the housing market model, is PPAD-complete \citep{papadimitriou1994complexity} even when the market consists of three agents with identical preferences. Moreover, we show that the query complexity of \housing{} with four or more agents is exponential in the approximation parameter. These hardness results stand in a stark contrast to the polynomial-time complexity of Top Trading Cycles algorithm for housing exchange and the Hungarian algorithm for housing assignment, and puts \housing{} in same complexity class as the computation of Arrow-Debreu equilibria \citep{chen2009settlingarrow} or Nash equilibria \citep{daskalakis2009complexity, chen2009settlingnash}.\footnote{See Part III in \citet{rubinstein2019hardness} for a superb overview of PPAD and for many results about important economic problems, such as computing an approximate Bayes-Nash equilibrium and an approximate competitive equilibrium from equal incomes, in this class.}

We study the complexity of the housing market problem in two computational models. In the \textit{white-box model}, also known as the polynomial function model, we assume that the preference functions of agents are given in the form of algorithms (or descriptions of Turing machines) with polynomial running time guarantees. We work in this model to show that \housing{} is PPAD-complete.
Secondly, we consider the \textit{black-box model} (or \textit{function oracle model}) in which the preference functions of agents are given by function oracles that can be queried to determine whether an agent demands a specific good at specific prices. The \textit{query complexity} is then defined as the number of queries needed to find a solution. Note that the \housing{} problem is at least as hard in the black-box model as in the white-box model, as the latter restricts the class of preferences to those expressible with polynomial-time functions. Moreover, the white-box model provides additional information about agent preferences through the algorithm implementing the preference functions. In the black-box model, we show exponential (in the approximation parameter) upper and lower bounds on the query complexity of \housing{} with $n$ agents.

The complexity class PPAD captures the computational complexity of many important problems that arise in economics and game theory. Two well-known PPAD-complete problems are \sperner{} and \cake{}. \sperner{} is the problem of finding a panchromatic triangle in a triangulation of a simplex or hypercube, while \cake{} captures the computational challenge of finding an envy-free allocation of a one-dimensional cake among players with heterogeneous preferences over pieces of cake. (See \cref{section:cake-cutting,section:sperner-to-rkkm,section:sperner-problem} for formal definitions of \cake{} and \sperner{}.)
To prove our main result that \housing{} is PPAD-complete, we first show that the problem is computationally equivalent to \rkkm{}, the generalization of the Knaster–Kuratowski–Mazurkiewicz lemma to multiple coverings due to \citet{gale1984equilibrium}. Our polynomial-time reductions between the two problems also imply that \housing{} and \rkkm{} have the same query complexity in the black-box model. We then develop polynomial-time reductions from two famous PPAD-complete problems (\cake{} and \spernerX{}) to \rkkm{}. Our reduction from \cake{} to \rkkm{} establishes lower bounds on the query complexity of \housing{}, as well as PPAD-hardness, for four or more agents. Surprisingly, our second reduction, from \spernerX{} to \rkkm{}, implies that \housing{} remains PPAD-hard even when the market consists of three agents and these agents have identical preferences. Finally, we show membership of \rkkm{} in PPAD by reducing \rkkm{} to \sperner{}. This reduction uses the technique of labeling and coloring the vertices of a triangulated simplex initially developed in \citet{su1999rental} to provide the first computational procedure for finding an envy-free cake-cutting solution; this technique was then refined in \citet{deng2012algorithmic} using Kuhn's triangulation to prove PPAD-hardness of \cake{}. Moreover, our reduction implies an exponential upper bound on the query complexity of \rkkm{}, which builds on a query complexity bound for \sperner{} shown in \citet{deng2011discrete}. Our results are summarized in \cref{table:results}, and the main polynomial-time reductions developed in this paper are depicted in \cref{fig:PPAD-reductions}.

\begin{figure}
\centering
\begin{tikzpicture}[xscale=1, yscale=0.8]
\tikzstyle{newproblem}=[draw, thick, rounded rectangle, fill=yellow!20, inner sep=8, outer sep=6]
\tikzstyle{knownproblem}=[draw, thick, rectangle, fill=blue!20, inner sep=8, outer sep=6]
\tikzstyle{reduction}=[->, ultra thick]
    \node[newproblem] (A) at (-6,0) {\housing{}};
    \node[knownproblem] (B) at (0,2.5) {\cake{}};
    \node[newproblem] (C) at (0,0) {\rkkm{}};
    \node[knownproblem] (D) at (0,-2.5) {\spernerX{}};
    \node[knownproblem] (E) at (6,0) {\sperner{}};
    \draw[reduction] (A) to[bend left=10] node[above, midway] {Thm~\ref{thm:housing-to-rkkm}} (C);
    \draw[reduction] (C) to[bend left=10] node[below, midway] {Thm~\ref{thm:rkkm-to-housing}} (A);
    \draw[reduction] (B) -- node[right, midway] {Thm~\ref{thm:cake-to-rkkm}} (C);
    \draw[reduction] (D) -- node[right, midway] {Thms~\ref{thm:sperner-to-kkm} \& \ref{thm:kkm-to-rkkm}} (C);
    \draw[reduction] (C) -- node[above, midway] {Thm~\ref{thm:rkkm-to-sperner}} (E);
\end{tikzpicture}
\caption{Illustration of the main reductions developed in this paper. The three problems shown in blue boxes are known to be PPAD-complete; \cake{} was shown to be PPAD-complete by \citep{deng2012algorithmic,hollender2023envy}, \spernerX{} in \citep{chen2009complexity} and \sperner{} is a canonical PPAD-complete problem introduced in \cite{papadimitriou1994complexity}. Arrows denote reductions from origin to destination.}
\label{fig:PPAD-reductions}
\end{figure}
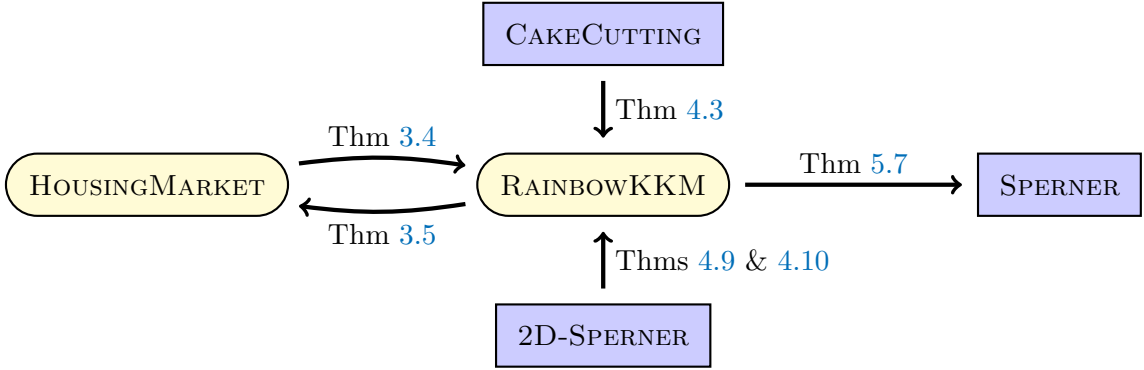

Our main result has implications for economic theory, computation and market design more broadly, as the housing market model is a key building block in understanding competitive equilibrium with indivisible goods in the presence of income effects \citep{baldwin2023equilibrium}. As the ability to compute equilibrium allocations is central to practical market design, our results suggest that more work is required to understand how to efficiently compute equilibrium allocation when budget constraints are a salient feature of the marketplace.

This paper relates to three strands of the literature.
Firstly, our results add to a growing body of work showing the PPAD-completeness of computing equilibria in markets \citep{chen2009settlingarrow,chen2009spending}. Our model, which guarantees existence of equilibrium with indivisible goods and income effects, is conceptually related to the work of \citet{chen2017complexity} which shows PPAD-hardness of solving Arrow-Debreu markets with non-monotone utilities (over divisible goods). It is also similar in spirit (although not in techniques) to the model of \citet{chen2022computational}, which studies equilibria in the classic unit-demand pseudomarket of \citet{hylland1979efficient}.  However, we prove our results in Gale's model of the housing market which is very general and allows for externalities and non-monotonicity of utility in money. It is an open question whether the PPAD-completeness survives with stronger assumptions on preferences (as, e.g., in the \citet{quinzii1984core} model).

Secondly, our reductions demonstrate parallels between \cake{} and the \housing{} and \rkkm{} problems. \cake{} was first shown to be PPAD-complete by \citet{deng2012algorithmic} in a model with demand functions. PPAD-completeness of \cake{} was then extended recently to the more stringent model with utility functions by \citet{hollender2023envy}. We make use of this latter cake cutting model when we reduce \cake{} to \housing{}. In Gale's model of the housing market, agent preferences are also represented by demand functions. It is an open question whether PPAD-hardness continues to hold for \housing{} when preferences are instead represented by utility functions. \cake{} has also been studied with externalities \citep{zhang2013externalities} and with respect to query complexity \citep{branzei2017query, hollender2023envy}.

Thirdly, our results are loosely related to a broader agenda on mechanism design with income effects (see, e.g.,  \citet{che1998standard,benoit2001multiple,bhattacharya2010budget,dobzinski2012multi,saitoh2008vickrey,morimoto2015strategy,baisa2017auction}). These papers typically consider the design of optimal mechanisms in the presence of income effects, whereas we focus on computation of competitive equilibrium prices and allocations.

\begin{table}[tb!]
    \centering
    \begin{tabular}{c|ccc}
        \toprule
        Complexity model    & $n=2$ & $n=3$ & $n \geq 4$\\
        \midrule
        white-box   & P     & PPAD-hard & PPAD-hard \\
        black-box   & $O(\log ( \frac{1}{\varepsilon}))$ & ? & $\Theta(\poly(\frac{1}{\varepsilon}))$\\
        \bottomrule
    \end{tabular}   
    \caption{The \textbf{computational complexity} (in the white-box model) and \textbf{query complexity} (in the black-box model) of the \housing{} problem with $n$ agents (and the computationally equivalent \rkkm{} problem). The two problems lie in PPAD for all $n$. The query complexity for $n=3$ is unknown but expected to be $\Theta(\poly(\frac{1}{\varepsilon}))$.}
    \label{table:results}
\end{table}

\paragraph{Organization.} The rest of the paper is organised as follows. At the end of this section, we introduce the PPAD complexity class and notation. \Cref{section:housing-and-rkkm} states the model, defines competitive equilibrium, introduces the Rainbow-KKM lemma used by \citet{gale1984equilibrium} to prove the  existence of competitive equilibrium in his model. In this section, we also provide initial reductions. \Cref{section:homeomorphism} establishes the computational equivalence of \housing{} and \rkkm{} in the white-box and black-box models. In \cref{section:ppadhardness}, we provide a reduction from \cake{} to \rkkm{} to derive a lower bound on the query complexity of \housing{} and \rkkm{}, and a reduction from \spernerX{} to \rkkm{} to establish PPAD-hardness with three identical agents. Section~\ref{section:membership} proves membership of \housing{} and \rkkm{} in PPAD, concluding our proof that the two problems are PPAD-complete, and gives an upper bound on the query complexity. Section~\ref{section:conclusion} is a conclusion.

\paragraph{The PPAD complexity class.}
Computational problems for which the existence of a solution is guaranteed are called \textit{total search problems}. The computational challenge is to find a solution. Examples of total search problems include integer factorisation, computing fixed points, and finding Nash equilibria. The class TFNP consists of all total search problems that can be solved in non-deterministic polynomial time; in other words, correctness of a solution to the problem can be verified in polynomial time. PPAD is the much-studied subclass of TFNP containing all problems that admit a polynomial-time reduction to the canonical problem \eol{} \citep{papadimitriou1994complexity}. For an introduction to PPAD and a definition of \eol{}, we refer to \citet{goldberg2011survey}.
A problem is PPAD-hard if \eol{} reduces to it, and PPAD-complete if it lies in PPAD and is PPAD-hard. PPAD-hard problems are believed not to be solvable in polynomial time.
As polynomial-time reductions are transitive, we can show membership of PPAD by reducing a problem to any other problem in PPAD, and PPAD-hardness by reducing from any problem that is PPAD-hard. Reductions between two problems also help us relate their query complexities.

\begin{definition}
For any two total search problems $\mathcal{P}$ and $\mathcal{Q}$, a \textit{polynomial-time reduction} from $\mathcal{P}$ to $\mathcal{Q}$ consists of two polynomial-time computable functions $f$ and $g$ such that $f$ maps any instance $x$ of $\mathcal{P}$ to an instance $f(x)$ of $\mathcal{Q}$, and $g$ maps any solution $s$ for $f(\xb)$ to a solution $g(s)$ for $\xb$.
\end{definition}

\paragraph{Notation.}
$\Z$ denotes the integers, $\N$ the positive natural numbers, and $\N_0$ the natural numbers including $0$. We write $[n] \coloneqq \{1, \ldots, n\}$ and $[n]_0 \coloneqq \{0,1,\ldots, n\}$. Let $\Delta_{n-1} \coloneqq \conv \{\eb^i \mid i \in [n] \}$ denote the standard $(n-1)$-simplex in $\R^n$. Throughout, the \textit{distance between two points $\xb, \yb \in \R^n$} is defined using the $L_1$ norm as $\|\xb - \yb \|_1 = \sum_{i \in [n]} |x_i - y_i|$. We write $\xb \geq \yb$ for two vectors $\xb$ and $\yb$ if the inequality holds elementwise. A function $f:X \to Y$ with $X,Y \subseteq \R^n$ is \emph{Lipschitz-continuous with constant $K$} (or \emph{$K$-Lipschitz}) if there exists a constant $K$ such that $\| f(\xb) - f(\yb) \|_1 \leq K \|\xb - \yb \|_1$ for all $\xb, \yb \in X$.

\section{The Housing Market and Rainbow-KKM problems}
\label{section:housing-and-rkkm}
We begin by introducing the \housing{} and \rkkm{} problems and developing initial results.

\subsection{The housing market}

Consider a market with $n$ agents $[n]$ that are each endowed with a house. Each house is identified with its agent, so the set of houses is $[n]$. The goal is to find a price for each house and an assignment of houses to agents so that every agent receives a house they prefer at these market prices.

The demand of each agent $i \in [n]$ is expressed using $n+1$ \emph{preference sets} $P^i = (P^i_0, \ldots, P^i_{n})$ that cover $\mathbb{R}^n$. At market prices $\pb \in \R^n$, agent $i$ demands good $j$ if $\pb \in P^i_j$, and nothing if $\pb \in R^i_0$. If the agent demands nothing, she prefers to exchange her house for money without obtaining a new house. Note that an agent may be indifferent between multiple houses (or between demanding nothing and demanding one or more houses) if preference sets overlap at market prices. \citet{gale1984equilibrium} imposes the following mild assumptions on the preference set of every agent $i \in [n]$. Let $B_n \coloneqq \{\bm{p} \in \mathbb{R}^n \mid \bm{p} \geq 0, p_j = 0 \text{ for some } j \in [n]\}$ be the boundary of the positive quadrant. For each agent $i$, we assume that

\begin{enumerate}[(i)]
    \item \label{assumption:closed-preference-sets} sets $P^i_0, P^i_1, \ldots, P^i_n$ are closed,
    \item \label{assumption:bounding-box} there exists some $M>0$ such that $\bm{p} \not \in P^i_j$ if $p_j \geq M$,
    \item \label{assumption:covering} $P^i_1, \ldots, P^i_{n}$ cover $\Sigma^M_{n-1} = \{\pb \in B_n \mid 0 \leq \pb \leq M \}$.
\end{enumerate}
Without loss of generality, we assume $M=1$ (re-scaling prices if necessary) and write $\Sigma_{n} \coloneqq \Sigma^M_{n}$.

Agents only wish to own at most one house, but otherwise the permissible preferences are very general: they allow for externalities among objects and demand to increase in prices (e.g.~Veblen goods). An \textit{assignment} is a permutation of $[n]$ that maps agents to houses. A \textit{market equilibrium} consists of market prices $\pb$ and an assignment $\pi$ such that every agent receives a house they demand at these prices; more formally, $\pb \in P^i_{\pi(i)}$ for every agent $i \in [n]$. An \textit{$\varepsilon$-equilibrium} is a pair $(\pb,\pi)$ such that $\pb$ is $\varepsilon$-close to $P^i_{\pi(i)}$ for every agent $i$, i.e.~there exists $\pb^i \in R^i_{\pi(i)}$ such that $\|\pb - \pb^i \|_1 \leq \varepsilon$ for every $i \in [n]$. \citet{gale1984equilibrium} establishes that assumptions \eqref{assumption:closed-preference-sets} to \eqref{assumption:covering} guarantee the existence of market equilibrium. 

\begin{theorem}[\citep{gale1984equilibrium}]
\label{thm:housing-equilibrium-existence}
    Under assumptions (i) to (iii), there exists a market equilibrium.
\end{theorem}

Hence finding an (approximate) equilibrium is a total search problem. As market prices can, in general, be irrational, we are interested in finding an $\varepsilon$-equilibrium for some approximation parameter $\varepsilon$. Moreover, we formalise preference sets in the computational setting by associating the preference sets $P^i = (P^i_1, \ldots, P^i_n)$ of each agent $i \in [n]$ with a \textit{preference function} $f^i:\R^n \times [n]_0 \to \{0,1\}$ so that $f^i(\pb,j) \coloneqq 1$ if $\pb \in C^i_j$ and $f^i(\pb,j) \coloneqq 0$ otherwise. We say that a preference function satisfies assumptions (i) to (iii) if its associated preference sets do so. This allows us to define the computational problem as follows.

\begin{problem}{\housing{}}
\textbf{Input:} An approximation parameter $\varepsilon \in (0, \frac{1}{4})$. Preference functions $f^1, \ldots, f^n$ satisfying assumptions (i) to (iii) for the agents $[n]$.

\textbf{Output:} An $\varepsilon$-equilibrium $(\pb,\pi)$.
\end{problem}
We refer to the computational problem with a fixed number $n$ of agents as $n$-\housing{}. In the black-box model, we assume that $f^1, \ldots, f^n$ are function oracles. In the white-box model, $f^1, \ldots, f^n$ are given as polynomial-time algorithms.

In \cref{thm:n-to-n+1}, we show that $n$-\housing{} reduces in polynomial time to $(n+1)$-\housing{} with the same approximation parameter. Starting with a market for $n$ agents, the reduction adds an additional agent $n+1$ and corresponding house $n+1$. Preferences are designed so that, for any equilibrium prices $\pb \in \R^{n+1}$ of the new market, agents $i \in [n]$ demand the same houses at $\pb$ in the new market as they do at $(p_1, \ldots, p_n)$ in the original market. Hence the prices and allocation restricted to $[n]$ form an equilibrium for the original market. It is also clear from the reduction in the proof of \cref{thm:n-to-n+1} that the two problems have the same query complexity.

\begin{theorem}
\label{thm:n-to-n+1}
For any $n \geq 1$, there exists a polynomial-time reduction from $n$-\housing{} to $(n+1)$-\housing{}.
\end{theorem}
\begin{proof}
Suppose $(\varepsilon, f^1, \ldots, f^n)$ is an instance of $n$-\housing{} associated with preference sets $P^1, \ldots, P^n$ for $n$ agents. We construct an instance of $(n+1)$-\housing{} with the same approximation parameter and preference sets $Q^1, \ldots, Q^{n+1}$ for $n+1$ agents as follows. For the first $n$ agents $i \in [n]$, we extend the old preference sets $P^i_j$ by defining 
\begin{equation}
\label{eq:preference-sets-construction}
\begin{aligned}
Q^i_0 &\coloneqq \R^n, \\
Q^i_j &\coloneqq \left \{ \qb \in \R^{n+1} \mid (q_1, \ldots, q_n) \in P^i_j \right \} \text{ for all } j \in [n], \\
Q^i_{n+1} &\coloneqq \left \{ \qb \in \R^{n+1} \mid q_{n+1} = 0 \right \}.
\end{aligned}
\end{equation}
For the last agent $n+1$, we define
\begin{equation}
\label{eq:preference-sets-construction-2}
\begin{aligned}
Q^{n+1}_0 &\coloneqq \R^n, \\
Q^{n+1}_j &\coloneqq \left \{ \qb \in \R^n \mid q_j = 0 \text{ and } q_{n+1} \geq \tfrac{3}{4} \right \} \text{ for all } j \in [n], \\
Q^{n+1}_{n+1} &\coloneqq \left \{ \qb \in \R^n \mid q_{n+1} \leq \tfrac{3}{4} \right \}.
\end{aligned}
\end{equation}

It is straightforward that the preference functions for $(Q^i_1, \ldots, Q^i_n)$ can be constructed efficiently and make at most one call to preference function $f^i$ for $(P^i_1, \ldots, P^i_n)$. We argue that these preference sets satisfy Gale's assumptions \eqref{assumption:closed-preference-sets} to \eqref{assumption:covering}. It is immediate that $Q^i_j$ is closed for all $i,j \in [n+1]$, so \eqref{assumption:closed-preference-sets} holds. Secondly, \eqref{assumption:bounding-box} follows from the fact that $P^i_j$ also satisfies \eqref{assumption:bounding-box} when $i,j \in [n]$, and is immediate for $i=n+1$ or $j = n+1$. Finally, we show that \eqref{assumption:covering} holds. Indeed, suppose $\qb \in \Sigma_{n+1}$. If $q_{n+1} = 0$, then $\qb \in Q^i_{n+1}$ for all agents $i \in [n+1]$. If $q_{n+1} > 0$, then $q_j = 0$ for some other house $j \in [n]$. For agents $i \in [n]$, this implies that $(q_1, \ldots, q_n) \in P^i_k$ for some $k \in [n]$ and so $\qb \in Q^i_k$, as the preference sets $(P^i_1, \ldots, P^i_n)$ satisfy $\eqref{assumption:covering}$. For agent $n+1$, we see that $\qb \in Q^{n+1}_{j}$ or $\qb \in Q^{n+1}_{n+1}$.

Now let $(\qb,\tau)$ be an $\varepsilon$-equilibrium of the new market with $n+1$ agents. By definition, there exists $\qb^i \in Q^i_{\tau(i)}$ with $\| \qb - \qb^i\|_1 \leq \varepsilon$ for every $i \in [n+1]$. We argue that $\tau$ assigns house $n+1$ to agent $n+1$. Suppose not, so $\tau(n+1) < n+1$ and $\tau(i) = n+1$ for some agent $i \in [n]$. By construction of $Q^i_{n+1}$ we have $q^i_{n+1} = 0$. It follows that $q^{n+1}_{n+1} \leq 2 \varepsilon \leq \frac{1}{2}$. But then $\qb^{n+1} \notin Q^{n+1}_{\tau(n+1)}$, a contradiction. Define the permutation $\pi$ of $[n]$ as the restriction of $\tau$ to $[n]$ and $\pb \coloneqq (q_1, \ldots, q_n)$. We now show that $(\pb,\pi)$ is an $\varepsilon$-equilibrium for the original market with $n$ agents. Let $\pb^i \coloneqq (q^i_1, \ldots, q^i_n)$ for each $i \in [n]$. We see that $\qb^i \in Q^i_{\tau(i)}$ implies $\pb^i \in P^i_{\pi(i)}$ by construction of $Q^i_j$. Moreover, $\| \pb - \pb^i \|_1 \leq \| \qb - \qb^i\| \leq \varepsilon$ for all $i \in [n]$.
\end{proof}

\subsection{The \kkm{} and \rkkm{} problems}
\label{section:rkkm}
We now introduce the KKM and Rainbow-KKM lemmas, as well as their corresponding computational problems. The \rkkm{} problem will serve as our intermediary problem to show that \housing{} is PPAD-complete.
Recall the standard $(n-1)$-simplex $\Delta_{n-1} \coloneqq \conv \{\eb^i \mid i \in [n] \}$. For any $S \subseteq [n]$, we let $F_S \coloneqq \conv \{\eb^i \mid i \in S\}$ be the face of $\Delta_{n-1}$ spanned by vectors $\eb^i$ for all $i \in S$. A \emph{KKM covering of $\Delta_{n-1}$} is a collection of $n$ closed subsets $C_1, \ldots, C_n$ of $\R^n$ such that $F_S \subseteq \bigcup_{i \in S} C_i$ for every $S \subseteq [n]$.
The Knaster-Kuratowski-Mazurkiewicz (KKM) lemma states that every KKM covering $(C_1, \ldots, C_n)$ of $\Delta_{n-1}$ admits a point that is contained in all $C_i$, $i \in [n]$. For the proof of \cref{thm:housing-equilibrium-existence}, \citet{gale1984equilibrium} extends this result to families of KKM coverings.

\begin{lemma}[\citep{knaster1929beweis}]
\label{lemma:KKM-point}
    Let $(C_1, \ldots, C_n)$ be a KKM covering of $\Delta_{n-1}$. There exists a point $\xb \in \Delta_{n-1}$ that lies in $\bigcap_{i \in [n]}C_i$.
\end{lemma}

\begin{lemma}[\citep{gale1984equilibrium}]
\label{lemma:RKKM-point}
Let $C^1, \ldots, C^n$ be $n$ KKM coverings of $\Delta_{n-1}$. There exists a permutation $\pi$ of $[n]$ and a point $\xb \in \Delta_{n-1}$ such that $\xb \in C^i_{\pi(i)}$ for every $i \in [n]$.
\end{lemma}

The proofs of \cref{lemma:KKM-point,lemma:RKKM-point} do not lead to an efficient way to compute the point $\xb$ in question. Moreover, as $\xb$ can be irrational, we consider the computational problems of finding an approximate solution. To formalize KKM coverings in our computational models, we associate a KKM covering $(C_1, \ldots, C_n)$ of $\Delta_{n-1}$ with a \textit{KKM covering function $g:\R^n \times [n] \to \{0,1\}$} by defining $g(\xb,j) \coloneqq 1$ if $\xb \in C^i_j$ and $g^i(\xb,j) \coloneqq 0$ otherwise. This leads to the following computational problems in which the covering functions are given as function oracles or as polynomial-time algorithms.

\begin{problem}{\kkm{}}
\textbf{Input:} Approximation parameter $\varepsilon \in (0, \frac{1}{4})$. Dimension $n$. KKM covering function $g$ of $\Delta_{n-1}$.

\textbf{Output:} Point $\xb$ for which $\| \xb - \xb^i \|_1 \leq \varepsilon$ for some $\xb^i \in C_i$ for every $i \in [n]$.
\end{problem}
We say that a pair $(\xb, \pi)$ consisting of point $\xb \in \Delta_{n-1}$ and permutation $\pi$ of $[n]$ is an \textit{$\varepsilon$-approximate Rainbow-KKM solution} if there exists $\xb^i \in C^i_{\pi(i)}$ with $\|\xb - \xb^i\|_1 \leq \varepsilon$ for every $i \in [n]$. We refer to the \kkm{} problem with fixed dimension $n=3$ as 3D-\kkm{}.

\begin{problem}{\rkkm{}}
\textbf{Input:} Approximation parameter $ \varepsilon \in (0, \frac{1}{4})$. KKM covering functions $g^1, \ldots, g^n$ of $\Delta_{n-1}$.

\textbf{Output:} $\varepsilon$-approximate Rainbow-KKM solution $(\xb,\pi)$.
\end{problem}
We refer to the \rkkm{} problem with fixed dimension $n$ as $n$-\rkkm{}. As we will see in \cref{section:sperner-to-rkkm}, finding an approximate Rainbow-KKM point is at least as hard as finding an approximate KKM point.

\subsection{Sparse KKM coverings}
In \cref{section:membership}, we reduce \rkkm{} to the problem \sperner{}. In the reduction, we assume that each KKM covering $(C_1, \ldots, C_n)$ of $\Delta_{n-1}$ satisfies $C_i \cap F_{[n] \setminus \{i\}} = \emptyset$ for every $i \in [n]$. We call such a KKM covering \textit{sparse}. We now argue that this is without loss of generality, as we can reduce any \rkkm{} instance to an instance with sparse KKM coverings. Note that we cannot simply make a KKM covering $(C_1, \ldots, C_n)$ sparse by subtracting $F_{[n] \setminus \{i\}}$ from $C_i$ for each $i \in [n]$, as the resulting sets are not closed.

Suppose $(C_1, \ldots, C_n)$ is a KKM covering of $\Delta_{n-1}$. We now construct a KKM covering $(D_1, \ldots, D_n)$ parametrised by $\delta > 0$ that is sparse. Moreover, we will see that for every $\xb \in D_i$ there exists a point $\yb \in C_i$ close to $\xb$. This forms the basis of our reduction from \rkkm{} to the version with sparse KKM coverings.

Fix parameter $\delta \in (0,\frac{1}{4})$. First we define a projection function $\tau: \Delta_{n-1} \to \Delta_{n-1}$. Let $\tau(\xb)$ be the vector obtained from $\xb$ by reducing all entries less or equal to $\delta$ down to $0$, and normalizing the remaining entries so that the entries sum to $1$. We write
\[
\hat{\tau}(\xb)_i \coloneqq
\begin{cases}
    0       &\text{if } x_i \leq \delta,\\
    x_i     & \text{else},
\end{cases}
\]
and $\tau(\xb) = \frac{1}{\| \hat{\tau}(\xb) \|_1}\hat{\tau}(\xb)$. Note that $\|\tau(\xb) - \xb \|_1 \leq 2n \delta$ for any $\xb \in \Delta_{n-1}$. We define $D_i$ for each $i \in [n]$ as
\begin{equation}
\label{eq:sparse-covering-construction}
D_i \coloneqq \{ \xb \in \Delta_{n-1} \mid x_i \geq \delta \text{ and } \tau(\xb) \in C_i \}.
\end{equation}

We now argue that $(D_1, \ldots, D_n)$ is a sparse KKM covering of $\Delta_{n-1}$. To see that each $D_i$ is closed, write $D_i$ as the union of sets $D^T_i \coloneqq \{\xb \in D_i \mid \{j \in [n] \mid x_j \leq \delta \} = T \}$ for all $T \subseteq [n] \setminus \{i\}$. By a standard limiting argument and the fact that $C_i$ is closed, we see that each $D^T_i$ is closed, and so their union $D_i$ is also closed. Next, we see that $(D_1, \ldots, D_n)$ satisfies the KKM covering property. Fix $S \subseteq [n]$ and $\xb \in F_S$. We let $S' = \{ i \in [n] \mid \tau(\xb)_i > 0 \}$, so $S' \subseteq S$ and $\tau(\xb) \in F_{S'}$. As $(C_1, \ldots, C_n)$ is a KKM covering, $\tau(\xb) \in \bigcup_{i \in S'} C_i$. Fix $j \in S'$ with $\tau(\xb) \in C_j$. Note that $\tau(\xb)_j > 0$ implies $\tau(\xb)_j > \delta$. So, by construction, $\tau(\xb) \in D_j$. Hence $\tau(\xb) \in \bigcup_{i \in S} D_i$, which concludes the proof that $(D_1, \ldots, D_n)$ is a KKM covering. Finally, as $D_i$ contains only points $\xb$ satisfying $x_i \geq \delta$ and $\xb \in D_{[n] \setminus \{i\}}$ implies $x_i = 0$, we see that $(D_1, \ldots, D_n)$ is sparse.

\begin{proposition}
\label{proposition:sparse-KKM-covering}
There exists a polynomial-time reduction from $n$-\rkkm{} to the problem $n-\rkkm{}$ with sparse coverings.
\end{proposition}
\begin{proof}
Suppose $(\varepsilon, g^1, \ldots, g^n)$ is an instance of $n$-\rkkm{} associated with KKM coverings $C^1, \ldots, C^n$. We construct an instance $(\varepsilon' = \frac{\varepsilon}{2}, h^1, \ldots, h^n)$ of $n$-\rkkm{} associated with sparse KKM coverings $D^1, \ldots, D^n$ that are constructed as described in \eqref{eq:sparse-covering-construction} with parameter $\delta = \frac{\varepsilon}{8n}$. It is straightforward that $h^i$ can be implemented efficiently with at most one call to $g^i$.

Suppose $(\xb,\pi)$ is a solution for instance $(\varepsilon', h^1, \ldots, h^n)$. Hence there exists $\xb^i \in D^i_{\pi(i)}$ with $\| \xb - \xb^i \|_1 \leq \varepsilon'$ for every $i \in [n]$. By construction of the sets $D^i_j$ in \eqref{eq:sparse-covering-construction}, we have $\tau(\xb^i) \in C^i_{\pi(i)}$. Moreover, we know that $\| \xb - \tau(\xb) \|_1 \leq 2n \delta$ and $\| \xb^i - \tau(\xb^i) \|_1 \leq 2n \delta$, so $\| \tau(\xb) - \tau(\xb^i) \|_1 \leq \varepsilon' + 4n \delta = \varepsilon$. Hence $(\tau(\xb), \pi)$ is a solution for the original instance $(\varepsilon, g^1, \ldots, g^n)$.
\end{proof}

\section{Reductions between \housing{} and \rkkm{}}
\label{section:homeomorphism}
We now establish the computational equivalence of \housing{} and \rkkm{} by proving polynomial-time reductions in both directions. We note that our reductions preserve the dimensionality of the problems; that is, an instance of $n$-\housing{} is reduced to an instance of $n$-\rkkm{}, and vice versa. Moreover, it follows from our reductions that $n$-\housing{} and $n$-\rkkm{} have the same query complexity.

In \cref{section:connecting-domains}, we define a homeomorphism $\phi : \Sigma_{n} \to \Delta_{n-1}$ from the `relevant' domain $\Sigma_{n}$ of the housing market to the domain $\Delta_{n-1}$ of the Rainbow-KKM problem. This homeomorphism was first introduced by \citet{gale1984equilibrium} to prove equilibrium existence for the housing market via the Rainbow-KKM lemma. \Cref{section:equivalence-reductions} then provides the reductions between \housing{} and \rkkm{}.

\subsection{Connecting the domains of \housing{} and \rkkm{}}
\label{section:connecting-domains}
Recall that $\Sigma_{n} = \{ \xb \in \R^n \mid 0 \leq \xb \leq 1, \text{ and } x_i = 0 \text{ for some } i \in [n]\}$ is the intersection of the unit cube with the boundary of the positive orthant $B_n$.
For each permutation $\pi$ of $[n]$, define the two sets
\begin{align*}
    \Sigma_\pi &\coloneqq \{\pb \in \Sigma_n \mid p_{\pi(1)} \geq p_{\pi(2)} \geq \cdots \geq p_{\pi(n)} = 0 \}\\
    \text{and }\Delta_{\pi} &\coloneqq \{\xb \in \Delta_{n-1} \mid x_{\pi(1)} \leq x_{\pi(2)} \leq \cdots \leq x_{\pi(n)} \}.
\end{align*}
It is immediate that $\Sigma_n = \bigcup_{\pi} \Sigma_{\pi}$ and $\Delta_{n-1} = \bigcup_{\pi} \Delta_{\pi}$. For each $\pb \in \Sigma_\pi$, we define
\begin{equation}
\label{eq:phi}
\phi(\pb)_{\pi(k)} \coloneqq
                      \frac{1 - p_{\pi(1)}}{n}
                    + \frac{p_{\pi(1)} - p_{\pi(2)}}{n-1}
                    + \frac{p_{\pi(k-1)} - p_{\pi(k)}}{n-k+1}
\end{equation}
The inverse $\phi^{-1} : \Delta_{n-1} \to \Sigma_n$ of $\phi$ is defined, for every $\xb \in \Delta_{\pi}$, as
\begin{equation}
\label{eq:phi-inverse}
\phi^{-1}(\xb)_{\pi(k)} \coloneqq 1 - x_{\pi(1)} - x_{\pi(2)} - \cdots - x_{\pi(k-1)} - (n-k+1)x_{\pi(k)}.
\end{equation}

\Cref{lemma:phi-properties} shows that $\phi^{-1}$ is indeed the inverse of $\phi$. Moreover, \cref{lemma:phi-lipschitz,lemma:phi-inverse-lipschitz} prove that $\phi$ is $n$-Lipschitz and $\phi^{-1}$ is $n^2$-Lipschitz. In particular, this establishes that $\phi$ is a homeomorphism.

\begin{lemma}
\label{lemma:phi-properties}
The function $\phi$ is a bijection from $\Sigma_n$ to $\Delta_{n-1}$, and $\phi^{-1}$ is its inverse. Moreover, $\phi(\Delta_{\pi}) = \Sigma_{\pi}$ for every permutation $\pi$ of $[n]$.
\end{lemma}
\begin{proof}
It is straightforward to check that $\phi^{-1}(\phi(\pb)) = \pb$ for every $\pb \in \Sigma_{n}$ and $\phi(\phi^{-1}(\xb)) = \xb$ for every $\xb \in \Delta_{n-1}$, so $\phi^{-1}$ is the inverse of $\phi$.

Next we show that $\phi(\Sigma_{\pi}) = \Delta_{\pi}$. Fix $\pb \in \Sigma_\pi$. By construction of $\Sigma_{\pi}$, we have $p_{\pi(k-1)} - p_{\pi(k)} \geq 0$ for all $k \in \{2, \ldots, n\}$. It follows from \eqref{eq:phi} that $\phi(\pb)_{\pi(1)} \geq 0$ and $\phi(\pb)_{\pi(k-1)} \leq \phi(\pb)_{\pi(k)}$ for all $k \in \{2, \ldots, n\}$, and
\[
\sum_{i \in [n]} \phi(\pb)_{\pi(i)} = n \cdot \frac{1-p_{\pi(1)}}{n} + (n-1) \cdot \frac{p_{\pi(1)} - p_{\pi(2)}}{n-1} + \cdots + 1 \cdot \frac{p_{\pi(n-1)} - p_{\pi(n)}}{1} = 1 - p_{\pi(n)} = 1.
\]
Hence, $\phi(\pb) \in \Delta_{\pi}$. As $\pb$ was chosen arbitrarily, we get $\phi(\Sigma_\pi) \subseteq \Delta_{\pi}$. To show that $\phi(\Sigma_\pi) = \Delta_{\pi}$, we now fix $\xb \in \Delta_{\pi}$ and show that $\phi^{-1}(\xb) \in \Sigma_{\pi}$. Let $\pb \coloneqq \phi^{-1}(\xb)$. By \eqref{eq:phi-inverse}, we have $p_{\pi(k-1)} - p_{\pi(k)} = (n - k +1) (x_{\pi(k-1)} - x_{\pi(k)}) \leq 0$, and so $p_{\pi(k-1)} \geq p_{\pi(k)}$, for every $k \in \{2,\ldots, n\}$. Moreover, $\xb \in \Delta_{\pi}$ tells us that $x_1 \leq \frac{1}{n}$ and $\sum_{i \in [n]}x_{\pi(i)} = 1$. Hence, $p_{\pi(1)} = 1 - n x_{\pi(1)} \in [0,1]$ and $p_{\pi(n)} = 1 - \sum_{i \in [n]} x_{\pi(i)} = 0$.
\end{proof}

\begin{lemma}
\label{lemma:phi-lipschitz}
The function $\phi$ is $n$-Lipschitz.
\end{lemma}
\begin{proof}
In order to prove that $\phi$ is $n$-Lipschitz, we need to show $\|\phi(\pb) - \phi(\qb)\|_1 \leq n \|\pb - \qb \|_1$ for all $\pb, \qb \in \Sigma_n$. Fix two points $\pb, \qb \in \Sigma_n$, and let $\|\pb-\qb\|_1 = \delta$.
If $\pb$ and $\qb$ both lie in $\Sigma_\pi$ for some permutation $\pi$, then
\begin{equation}
\label{eq:p-q-same-simplex}
| \phi(\pb)_{\pi(k)} - \phi(\qb)_{\pi(k)} | = 
\sum_{j \in [k]} \left( \frac{1}{n-j+1} - \frac{1}{n-j+2} \right ) | p_{\pi(j)} - q_{\pi(j)} |
\leq \sum_{j \in [k]} | p_{\pi(j)} - q_{\pi(j)} | \leq \delta
\end{equation}
and so $\|\phi(\pb) - \phi(\qb)\|_1 \leq n \delta$.

From now on, we assume that $\pb \in \Sigma_{\pi}$ and $\qb \notin \Sigma_{\pi}$. Suppose first that $\pb$ and $\qb$ differ only in one entry, so $\qb = \pb + \alpha \eb^{\pi(j)}$ for some $j \in [n]$. Without loss of generality, we can assume that $\alpha > 0$ (by relabeling $\pb$ and $\qb$ if $\alpha < 0$). Note that $\qb \notin \Sigma_{\pi}$ implies $j \geq 2$. Moreover, increasing $p_{\pi(j)}$ by $\alpha$ moves it to the position of $p_{\pi(i)}$ for some $i < j$, and moves entries $p_{\pi(i)}, \ldots, p_{\pi(j-1)}$ one position to the right.  The entries of $\pb$ and $\qb$ in descending order are, thus, given by
\begin{align}
\label{eq:sorted-p-and-q}
    &p_{\pi(1)}, p_{\pi(2)}, \ldots, p_{\pi(i-1)}, p_{\pi(i)}, \ldots, p_{\pi(j-1)}, p_{\pi(j)}, p_{\pi(j+1)}, \ldots, p_{\pi(n)}, \\
    &p_{\pi(1)}, p_{\pi(2)}, \ldots, p_{\pi(i-1)}, p_{\pi(j)}+\alpha, p_{\pi(i)}, \ldots, p_{\pi(j-1)}, p_{\pi(j+1)}, \ldots, p_{\pi(n)}.
\end{align}
Observe that 
\begin{equation}
\label{eq:diameter}
p_{\pi(j)} \leq p_{\pi(k)} \leq p_{\pi(j)} + \alpha, \ \forall k \in \{i, \ldots, j\}.
\end{equation}

We will now show that $|\phi(\pb)_{\pi(k)} - \phi(\qb)_{\pi(k)}| \leq \alpha$ for all $k \in [n]$. Fix some $k \in [n]$ and note
\[
\phi(\pb)_{\pi(k)} = \frac{1-p_{\pi(1)}}{n} + \sum_{l=2}^k \frac{p_{\pi(l-1)} - p_{\pi(l)}}{n-l+1}.
\]

If $k \in \{1, \ldots, i-1\}$, then $\phi(\pb)_{\pi(k)} = \phi(\qb)_{\pi(k)}$ by \eqref{eq:sorted-p-and-q}. Now suppose that $k \in \{i, \ldots, j-1\}$. Then
\begin{align*}
\phi(\qb)_{\pi(k)}
= &\frac{1-p_{\pi(1)}}{n} + \sum_{l=2}^{i-1} \frac{p_{\pi(l-1)} - p_{\pi(l)}}{n-l+1} \\
    + &\frac{p_{\pi(i-1)} - (p_{\pi(j)}+\alpha)}{n-i+1}
    + \frac{(p_{\pi(j)} + \alpha) - p_{\pi(i)}}{n-i+2} \\
    + &\sum_{l=i+1}^k \frac{p_{\pi(l-1)} - p_{\pi(l)}}{n-l+2},
\end{align*}
and so
\begin{align*}
|\phi(\pb)_{\pi(k)} - \phi(\qb)_{\pi(k)}|
    &= \frac{p_{\pi(i-1)} - p_{\pi(i)}}{n-i+1}
        - \frac{p_{\pi(i-1)} - (p_{\pi(j)} + \alpha)}{n-i+1}
        - \frac{p_{\pi(j)}+\alpha - p_{\pi(i)}}{n-i+2}\\
        &+ \sum_{l=i+1}^k \left (\frac{1}{n-i+1} - \frac{1}{n-i+2} \right ) (p_{\pi(l-1)} - p_{\pi(l)})\\
    &\leq \left (\frac{1}{n-i+1} - \frac{1}{n-i+2} \right )(p_{\pi(j)}+\alpha - p_{\pi(i)} + p_{\pi(i)} - p_{\pi(k)})\\
    &\leq p_{\pi(j)}+\alpha - p_{\pi(k)} \leq \alpha.
\end{align*}
The last inequality follows from \eqref{eq:diameter}.
Next, suppose $k=j$. In this case, we have
\[
\phi(\qb)_{\pi(j)} = \frac{1-p_{\pi(1)}}{n} + \sum_{l=2}^{i-1} \frac{p_{\pi(l-1)}-p_{\pi(l)}}{n-l+1} + \frac{p_{\pi(i-1)} - (p_{\pi(j)} + \alpha)}{n-i+1},
\]
so
\begin{align*}
|\phi(\pb)_{\pi(j)} - \pi(\qb)_{\pi(j)}|
    &= \sum_{l=i}^j \frac{p_{\pi(l-1)} - p_{\pi(l)}}{n-l+1} - \frac{p_{\pi(i-1)} - (p_{\pi(j)} + \alpha)}{n-i+1} \\
    &= \frac{p_{\pi(j)} + \alpha - p_{\pi(i)}}{n-i+1} + \sum_{l=i+1}^j \frac{p_{\pi(l-1)} - p_{\pi(l)}}{n-l+1}\\
    &\leq p_{\pi(j)} + \alpha - p_{\pi(i)} + p_{\pi(i)} - p_{\pi(j)} = \alpha.
\end{align*}
Now suppose that $k \in \{j+1, \ldots, n\}$. Hence,
\begin{align*}
\phi(\qb)_{\pi(k)} 
    &= \frac{1-p_{\pi(1)}}{n} + \sum_{l=2}^{i-1} \frac{p_{\pi(l-1)}-p_{\pi(l)}}{n-l+1} \\
    &+ \frac{p_{\pi(i-1)} - (p_{\pi(j)} + \alpha)}{n-i+1} + \frac{p_{\pi(j)} + \alpha - p_{\pi(i)}}{n-i+2} \\
    &+ \sum_{l=i}^{j-1} \frac{p_{\pi(l-1)}-p_{\pi(l)}}{n-l+1} + \frac{p_{\pi(j-1)} - p_{\pi(j+1)}}{n-j+2}\\
    &+ \sum_{l=j+2}^{k} \frac{p_{\pi(l-1)}-p_{\pi(l)}}{n-l+1}.
\end{align*}
So we have
\begin{align*}
|\phi(\pb)_{\pi(k)} - \phi(\qb)_{\pi(k)}|
    &= \frac{p_{\pi(i-1)} - p_{\pi(i)}}{n-i+1} - \frac{p_{\pi(i-1)} - p_{\pi(j)}+\alpha}{n-i+1} - \frac{p_{\pi(j)} + \alpha - p_{\pi(i)}}{n-i+2} \\
    &+ \sum_{l=i+1}^{j-1} \left ( \frac{1}{n-j+1} - \frac{1}{n-j+2} \right ) (p_{\pi(l-1)} - p_{\pi(l)}) \\
    &+ \frac{p_{\pi(j-1)} - p_{\pi(j)}}{n-j+1} - \frac{p_{\pi(j)} - p_{\pi(j+1)}}{n-i+2} - \frac{p_{\pi(j-1)} - p_{\pi(j+1)}}{n-i+2}\\
    &= \left ( \frac{1}{n-j+1} - \frac{1}{n-j+2} \right ) ( p_{\pi(j)} + \alpha - p_{\pi(i)}) \\
    &+ \left ( \frac{1}{n-j+1} - \frac{1}{n-j+2} \right ) (p_{\pi(i)} - p_{\pi(j-1)}) \\
    &+ \left ( \frac{1}{n-j+1} - \frac{1}{n-j+2} \right ) (p_{\pi(j-1)}- p_{\pi(j)}) \\
    &\leq \alpha.
\end{align*}

Finally, suppose $\pb$ and $\qb$ differ in arbitrarily many entries, so $\qb = \pb + \sum_{i \in [n]} \alpha_i \eb^i$ for some $\alpha \in \R^n$ and we have $\delta = \| \bm{\alpha} \|_1$. Define $\qb^{(k)} = \pb + \sum_{i=1}^k \alpha_i \eb^i$. (So $\qb^{(0)} = \pb$ and $\qb^{(n)} = \qb$.) Then
\[
\| \phi(\pb) - \phi(\qb) \|_1 = \left \| \sum_{i=1}^n \phi(\qb^{(i-1)}) - \phi(\qb^{(i)}) \right \|_1 \leq \sum_{i=1}^n \left \| \phi(\qb^{(i-1)}) - \phi(\qb^{(i)}) \right \|_1  \leq n \sum_{i=1}^n |\alpha_i| \leq n \delta.
\]
The first inequality makes use of the triangle inequality. The second inequality applies the result from the previous paragraph to each $\|\phi(\pb^{(i-1)}) - \phi(\qb^{(i)})\|_1$, as $\qb^{(i-1)}$ and $\qb^{(i)}$ differ in one entry, by $\alpha_i$.
\end{proof}

\begin{lemma}
\label{lemma:phi-inverse-lipschitz}
The function $\phi^{-1}$ is $n^2$-Lipschitz.
\end{lemma}
\begin{proof}
Define $\R^n_{\pi} = \{\xb \in \R^n \mid x_{\pi(1)} \leq x_{\pi(2)} \leq \cdots \leq x_{\pi(n)} \}$ for each permutation $\pi$ of $[n]$. For the purpose of this proof, we extend $\phi^{-1}$ to domain $\R^n$ by defining $\phi^{-1}(\xb)_{\pi(k)}$ as in \eqref{eq:phi-inverse} for every $\xb \in \R^n_{\pi}$.
We will show that $\|\phi^{-1}(\xb) - \phi^{-1}(\yb) \| \leq n^2 \|\xb - \yb \|_1$ for all $\xb, \yb \in \R^n$. The proof is conceptually similar to that of \cref{lemma:phi-lipschitz}. Fix $\xb, \yb \in \R^n$, and define $\delta \coloneqq \|\xb - \yb \|_1$. If $\xb$ and $\yb$ both lie in $\R^n_{\pi}$ for some permutation $\pi$, then, by applying the definition \eqref{eq:phi-inverse} of $\phi^{-1}$, we have
\[
|\phi^{-1}(\xb)_{\pi(k)} - \phi^{-1}(\yb)_{\pi(k)}| \leq (n-k+1)|x_{\pi(k)} - y_{\pi(k)}| + \sum_{i=1}^{k-1} |x_{\pi(i)} - y_{\pi(i)}| \leq (n-k+1) \delta + (k-1)\delta = n \delta,
\]
so $\| \phi^{-1}(\xb) - \phi^{-1}(\yb) \|_1 \leq n^2 \delta$.

From now on, we assume that $\xb \in \R^n_\pi$ and $\yb \notin \R^n_\pi$. Suppose first that $\pb$ and $\qb$ differ only in one entry, so $\qb = \pb - \alpha \eb^{\pi(j)}$ for some $j \in [n]$. Without loss of generality, assume $\alpha > 0$ (relabeling $\pb$ and $\qb$ if necessary). Note that $\qb \notin \R^n_{\pi}$ implies $j \geq 2$. Decreasing entry $x_{\pi(j)}$ by $\alpha$ moves it to the position of $x_{\pi(i)}$ for some $i < j$ (strict inequality holds due to $\yb \notin \R^n_{\pi}$), and moves entries $x_{\pi(i)}, \ldots, x_{\pi(j-1)}$ one position to the right. The entries of $\yb$ in ascending order are, thus,

\begin{align}
\label{eq:sorted-x-and-y}
    &x_{\pi(1)}, x_{\pi(2)}, \ldots, x_{\pi(i-1)}, x_{\pi(j)}-\alpha, x_{\pi(i)}, \ldots, x_{\pi(j-1)}, x_{\pi(j+1)}, \ldots, x_{\pi(n)}.
\end{align}
Observe that 
\begin{equation}
\label{eq:xy-diameter}
x_{\pi(j)} - \alpha \leq x_{\pi(k)} \leq x_{\pi(j)}, \ \forall k \in \{i, \ldots, j\}.
\end{equation}

We will now show that $|\phi^{-1}(\xb)_{\pi(k)} - \phi^{-1}(\yb)_{\pi(k)}| \leq n \alpha$ for all $k \in [n]$, which implies $\| \phi^{-1}(\xb) - \phi^{-1}(\yb) \|_1 \leq n^2 \delta$. Fix $k \in [n]$ and note that 
\[
\phi^{-1}(\xb)_{\pi(k)} = 1 - \sum_{l=1}^{k-1} x_{\pi(l)} - (n-k+1) x_{\pi(k)}.
\]
We distinguish between four cases.
If $k \in \{1, \ldots, i-1 \}$, then $\xb$ and $\yb$ coincide on entries $\pi(1), \ldots, \pi(k)$, so $\phi^{-1}(\xb)_{\pi(k)} = \phi(\yb)_{\pi(k)}$.
Secondly, suppose $k \in \{i, \ldots, j-1 \}$. Then $\phi^{-1}(\yb)_{\pi(k)} = 1 - \sum_{l=1}^{k-1} x_{\pi(l)} - (x_{\pi(j)} - \alpha) - (n-k) \xb_{\pi(k)}$, which implies $|\phi^{-1}(\xb)_{\pi(k)} - \phi^{-1}(\yb)_{\pi(k)}| \leq |x_{\pi(j)} - \alpha - x_{\pi(k)}| \leq \alpha$ by \eqref{eq:xy-diameter}.
Thirdly, suppose $k=j$. Then $\phi^{-1}(\yb)_{\pi(k)} = 1 - \sum_{l=1}^{i-1} x_{\pi(l)} - (n-i+1) (x_{\pi(j)}-\alpha)$, so $|\phi^{-1}(\xb)_{\pi(k)} - \phi^{-1}(\yb)_{\pi(k)}| \leq \left |\sum_{l=i}^{j-1} (x_{\pi(j)} - x_{\pi(l)}) - (n-i+1)\alpha \right | \leq n \alpha$, by \eqref{eq:xy-diameter}.
Lastly, suppose $k \in \{j+1, \ldots, n\}$. Then $\phi^{-1}(\yb)_{\pi(k)} = 1 + \alpha - \sum_{l=1}^{k-1} x_{\pi(l)} - (n-k+1) x_{\pi(k)}$, so $|\phi^{-1}(\xb)_{\pi(k)} - \phi^{-1}(\yb)_{\pi(k)}| = \alpha$.

Now suppose $\xb$ and $\yb$ differ in arbitrarily many entries. As in the previous proof, we have $\yb = \xb + \sum_{i=1}^n \alpha_i \eb^i$ for some $\bm{\alpha} \in \R^n$ and $\delta = \| \bm{\alpha} \|_1$. Define $\xb^{(k)} = \xb + \sum_{i=1}^k \alpha_i \eb^i$ for all $k \in \{0, \ldots, n\}$, so $\xb^{(0)} = \xb$ and $\xb^{(n)} = \yb$. Then $\xb^{(i-1)}$ and $\xb^{(i)}$ differ on one entry, by $\alpha_i$. Hence,
\[
\| \phi^{-1}(\xb) - \phi^{-1}(\yb)\|_1 \leq \sum_{i=1}^n \|\phi^{-1}(\xb^{(i-1)}) - \phi^{-1}(\xb^{(i)}) \|_1 \leq n^2 \sum_{i=1}^n |\alpha_i| \leq n^2 \delta.
\]
\end{proof}

\subsection{The computational equivalence}
We now provide polynomial-time reductions between \housing{} and \rkkm{}, and show that the two problems have the same query complexity in the black-box model.

\label{section:equivalence-reductions}
\begin{theorem}
\label{thm:housing-to-rkkm}
There exists a polynomial-time reduction from \housing{} to \rkkm{}.
\end{theorem}
\begin{proof}
Suppose $(\varepsilon, f^1, \ldots, f^n)$ is an instance of \housing{}, and each of the $n$ functions $f^i$ is associated with preference set $P^i = (P^i_1, \ldots, P^i_n)$. We construct an instance $(\varepsilon', g^1, \ldots, g^n)$ of $\rkkm{}$. Set the approximation parameter to $\varepsilon' = \frac{\varepsilon}{n^2}$, and let each $g^i$ be the covering function associated with the sparse KKM covering $C^i = (C^i_1, \ldots, C^i_n)$ of $\Delta_{n-1}$ defined by $C^i_j \coloneqq \phi(P^i_j \cap \Sigma_n)$ for every $i,j \in [n]$. As $\phi$ is a bijection, a point $\xb$ lies in $C^i_j$ if and only if $\phi^{-1}(\xb)$ lies in $P^i_j \cap \Sigma_n$. Thus, we can efficiently implement $g^i$ as a polynomial-time function which makes at most one call to $f^i$.

Next, we verify that the $(C^{i}_1, \ldots, C^i_n)$ are sparse KKM coverings of $\Delta_{n-1}$. Fix an agent $i$. It is immediate that the sets $C^i_j$ are closed, as the $P^i_j$ are closed and $\phi$ is a homeomorphism. We now check that $F_S \subseteq \bigcup_{j \in S} C^i_j$ for every $S \subseteq [n]$. Suppose that $S = [n]$. By assumption \eqref{assumption:covering} on preference sets, $P^i_1, \ldots, P^i_n$ are a covering of $\Sigma_n$ and $\phi$ is a bijection from $\Sigma_n$ to $\Delta_{n-1}$. Hence, $C^i_1, \ldots, C^i_n$ is a covering of $\Delta_{n-1}$. Now suppose $S \subsetneq [n]$ and let $\xb \in F_S$, so $x_j = 0$ for all $j \notin S$. Let $T = [n] \setminus S$. We pick a permutation $\pi$ of $[n]$ such that $T = \{\pi(1), \ldots, \pi(|T|)\}$ and $\xb \in \Delta_{\pi}$, so $x_{\pi(j)} = 0$ for all $j \in [|T|]$. Let $\pb = \phi^{-1}(\xb)$. By the definition of $\phi$ in \eqref{eq:phi-inverse}, we have $p_{\pi(j)} = 1$ for all $j \in [|T|]$. It follows by assumption \eqref{assumption:bounding-box} on preference sets that $\pb \notin \bigcup_{j \in [|T|]} P^i_{\pi(j)} = \bigcup_{k \in T} P^i_{k}$. As $(C^i_1, \ldots, C^i_n)$ is a covering of $\Delta_{n-1}$, we have $\phi(\pb) = \xb \in \bigcup_{k \in [n]} C^i_k$. It follows that $\xb \in \bigcup_{k \in S} C^i_k$. Finally, we use a similar argument to show that $(C^i_1, \ldots, C^i_n)$ is sparse, i.e.~$C^i_j \cap F_{[n] \setminus \{j\}} = \emptyset$ for every $i \in [n]$. Suppose $\xb \in C^i_j$, so $\pb = \phi^{-1}(\xb) \in P^i_j$. Suppose $\xb \in F_{[n] \setminus \{j\}}$, so $x_j = 0$. Letting $\pi$ be a permutation with $\pi(1) = j$, we see that $p_j = 1$ and so, by assumption (ii) of preference sets we have $\pb \notin P^i_j$, a contradiction to $\xb \in C^i_j$.
    
Let $(\xb, \pi)$ be an $\varepsilon'$-approximate Rainbow-KKM solution for our \rkkm{} instance. Thus there exists a point $\xb^i \in C^i_{\pi(i)}$ with $\| \xb - \xb^i \|_1 \leq \varepsilon'$ for every $i \in [n]$. Let $\pb = \phi^{-1}(\xb)$ and $\pb^i = \phi^{-1}(\xb^i)$. Firstly, $\xb^i \in C^i_{\pi(i)}$ implies $\pb^i \in P^i_{\pi(i)}$. Moreover, $\|\pb - \pb^i \|_1 = \|\phi^{-1}(\xb) - \phi^{-1}(\xb^i) \|_1 \leq n^2 \varepsilon' = \varepsilon$ for every $i \in [n]$, as $\phi^{-1}$ is $n^2$-Lipschitz. It follows that $(\pb, \pi)$ is an $\varepsilon$-equilibrium solution for the initial \housing{} instance.
\end{proof}

\begin{theorem}
\label{thm:rkkm-to-housing}
There exists a polynomial-time reduction from \rkkm{} to \housing{}.
\end{theorem}
\begin{proof}
Suppose $(\varepsilon, g^1, \ldots, g^n)$ is an instance of $\rkkm{}$, and each $g^i$ is associated with KKM covering $C^i_1, \ldots, C^i_n$ of $\Delta_{n-1}$. Without loss of generality, we can assume that the KKM coverings are sparse (cf.~\cref{section:rkkm}). We construct an instance $(\varepsilon', f^1, \ldots, f^n)$ of \housing{} with $\varepsilon' = \frac{\varepsilon}{n}$. Each preference function $f^i$ is associated with preference sets $P^i \coloneqq (P^i_0, \ldots, P^i_n)$ given by $P^i_0 = \R^n$ and $P^i_j = \phi^{-1}(C^i_j \cap \Delta_{n-1})$ for $j \in [n]$. As $\pb \in P^i_0$ for all $\pb \in \R^n$, and $\pb \in P^i_j$ if and only if $\phi(\pb) \in C^i_j \cap \Delta_{n-1}$ for any $j \in [n]$, we see that each $f^i$ can be implemented efficiently with at most one call to $g^i$.

It is immediate that $P^i$ covers $\R^n$, and that $P^i_0$ is closed. Moreover, by assumption the $C^i_j$ are closed and cover $\Delta_{n-1}$. Hence, as $\phi^{-1}$ is a homeomorphism from $\Delta$ to $\Sigma_{n}$, we see that the $P^i_j$ are closed and cover $\Sigma_n$. So $P^1, \ldots, P^n$ satisfy assumptions \eqref{assumption:closed-preference-sets} and \eqref{assumption:covering} on preference sets. Finally, we argue that $p_j \geq 1$ implies $\pb \notin P^i_j$ for every $j \in [n]$, in order to show that the preference sets satisfy assumption \eqref{assumption:bounding-box}. If $\pb \notin \Sigma_n$, then this follows from $P^i_j \subseteq \Sigma_n$. Now suppose $\pb \in \Sigma_n$ and $p_j = 1$. Then $\pb$ lies in $\Sigma_{\pi}$ for permutation $\pi = (j,\ldots)$ as defined in \cref{section:connecting-domains}. Hence $\xb \coloneqq \phi^{-1}(\pb)$ satisfies $x_j = 0$ and so $\xb \in F_{[n]\setminus\{j\}}$. As the KKM coverings are sparse by assumption, we get $\xb \notin C^i_j$.

Now let $(\pb, \pi)$ be an $\varepsilon'$-equilibrium solution for the \housing{} instance. Thus there exists $\pb^i \in P^i_{\pi(i)}$ with $\|\pb - \pb^i\|_1 \leq \varepsilon'$ for every $i \in [n]$. Note that $P^i_{\pi(i)} = \phi^{-1}(C^i_{\pi(i)} \cap \Delta_{n-1})$ is a subset of $\Sigma_n$ by definition of $\phi^{-1}$, so at least one entry of $\pb^i$ is $0$. It follows that $p_j \leq \varepsilon'$ for some $j \in [n]$. We claim that $(\xb, \pi)$ with $\xb = \phi(\pb)$ is an $\varepsilon$-approximate Rainbow-KKM solution for the original $\rkkm{}$ instance. We have $\phi(\pb^i) \in C^i_{\pi(i)}$. Moreover, as $\phi$ is $n$-Lipschitz by \cref{lemma:phi-lipschitz}, $\| \phi(\pb) - \phi(\pb^i) \|_1 \leq n \varepsilon' = \varepsilon$.
\end{proof}

\cref{thm:housing-to-rkkm,thm:rkkm-to-housing} establish that \housing{} is PPAD-complete (in the white-box model) if and only if \rkkm{} is. Moreover, in the proofs above we described how to implement the covering functions using one call to the preference functions, and vice versa. Hence \housing{} and \rkkm{} have the same query complexity in the black-box model.

\begin{corollary}
\label{corollary:housing-rkkm-query-complexity}
\housing{} and \rkkm{} have the same query complexity.
\end{corollary}

It also follows from the reduction of \cref{thm:housing-to-rkkm} that computing a market equilibrium in the housing market is straightforward when the market consists of two agents. Indeed, we briefly describe how an $\varepsilon$-approximate Rainbow-KKM solution can be found in polynomial time using straightforward binary search on the interval $\Delta_1$ when $n=2$. Suppose $(C^1_1, C^1_2)$ and $(C^2_1, C^2_2)$ are two KKM coverings of $\Delta_1 = \conv \{\eb^1, \eb^2 \}$, and initialize the search interval $[\xb,\yb]$ with $\xb = (0,1)$ and $\yb = (1,0)$. Note that $\xb \in C^1_1$ and $\yb \in C^2_2$, and $\|\xb - \yb \|_1 = 2$. Now repeat the following until $\|\xb - \yb \|_1 \leq \varepsilon$, at which point $\xb$ and $\yb$ are both $\varepsilon$-approximate Rainbow-KKM solutions. Compute the halfway point $\zb = \frac{1}{2}\xb + \frac{1}{2} \yb$. If $\zb$ does not lie in $C^1_1 \cup C^2_2$, then the KKM covering property guarantees $\zb \in C^1_2 \cap C^2_1$, and we have found an exact Rainbow-KKM solution. Otherwise, update the left boundary $\xb$ of the search interval to $\zb$ if $\zb \in C^1_1$ and the right boundary $\yb$ of the search interval to $\zb$ if $\zb \in C^2_2$. As $\|\xb - \yb \|_1$ reduces by at least half in every iteration and we maintain the invariants $\xb \in C^1_1$ and $\yb \in C^2_2$ throughout, this procedure terminates correctly with a running time and query complexity of $O(\log(\frac{1}{\varepsilon}))$.

\section{PPAD-Hardness of \housing{} and \rkkm}\label{section:ppadhardness}
We first prove that \housing{} is PPAD-hard by reducing \cake{} to the intermediate problem \rkkm{}. Our reduction shows that any algorithm for solving the housing market with a fixed number $d$ of agents can be applied to cutting a cake for $d$ players. It also leads to an exponential lower bound $\Omega(\poly (\frac{1}{\varepsilon}))$ on the query complexity of $d$-\rkkm{} and $d$-\housing{} in the black-box model for $d \geq 4$.

In \cref{section:sperner-to-rkkm}, we also provide a reduction from \spernerX{} to 3D-\kkm{}. As 3D-\kkm{} can be reduced to $3$-\rkkm{} by making multiple copies of the same KKM covering, this reduction implies that \housing{} is hard even if we have only three agents with identical preferences. \cake{} and \spernerX{} are known to be PPAD-complete problems \citep{deng2012algorithmic,hollender2023envy,chen2009complexity}.

\subsection{Reducing from the \cake{} problem}
\label{section:cake-cutting}
Consider a cake represented by one-dimensional interval $[0,1]$. The goal is to divide the cake among $d$ players by making $d-1$ cuts and assigning the resulting pieces to agents so that no agent envies another. The preferences of each agent $i \in [n]$ are described by a utility function $u^i$ defined on the set of all possible pieces $[a,b] \subseteq [0,1]$ (with $a \leq b$). These utility functions satisfy
\begin{enumerate}[(i)]
    \item \label{item:cake-1} $u^i(\emptyset) = 0$ and $u^i([a,b]) > 0$ for all $0 \leq a < b \leq 1$,
    \item \label{item:cake-2} $u^i([a,b]) - u^i([a',b']) \leq K (|a-a'| + |b-b'|)$ for any two intervals $[a,b], [a',b'] \subseteq [0,1]$ and some constant $K$.
\end{enumerate}
The first condition is known as \textit{non-negativity} and \textit{hungriness}, and the second condition as \textit{Lipschitz continuity (for intervals)}.

A $(d-1)$-cut dividing the cake $[0,1]$ into $d$ pieces can be represented by a vector $\bm{\xb} \in [0,1]^d$ with $\sum_{i \in [n]} x_i = 1$. Entry $x_k$ denotes the length of the $k$-th piece, and so the $k$-th piece of the cut is the interval $I_k(\xb) \coloneqq [\sum_{j \in [k-1]} x_j, \sum_{j \in [k]} x_j]$. It is easy to see that the space of all $(d-1)$-cuts consists of the simplex $\Delta_{d-1}$. A cut $\xb$ is \textit{envy-free} if there exists a one-to-one assignment of pieces to players so that each player (weakly) prefers their piece to all other pieces, i.e., there exists a permutation $\pi$ of $[d]$ such that $u^i(I_{\pi(i)}(\xb)) \geq u^i(I_k(\xb))$ for all $k \in [n]$. A cut is \textit{$\varepsilon$-envy-free} if $u^i(I_{\pi(i)}(\xb)) \geq u^i(I_k(\xb)) - \varepsilon$.

\begin{theorem}[\cite{stromquist1980how}]
Under assumptions \eqref{item:cake-1} and \eqref{item:cake-2}, there exists an envy-free cut.
\end{theorem}

In the corresponding total search problem \cake{} stated below, we assume that $u^1, \ldots, u^d$ are given as function oracles or polynomial-time algorithms, depending on the computational model. We let $d$-\cake{} denote the problem with a fixed number $d$ of players.

\begin{problem}{\cake{}}
\textbf{Input:} Approximation parameter $\varepsilon > 0$. Utility functions $u^1, \ldots, u^d$ for agents $[d]$ satisfying assumptions \eqref{item:cake-1} and \eqref{item:cake-2}.

\textbf{Output:} $\varepsilon$-envy-free cut $\xb \in \Delta_{d-1}$ and corresponding assignment $\pi$ of pieces to players.
\end{problem}

\begin{theorem}[\citep{hollender2023envy}]
\label{thm:cake-is-PPAD-complete}
\cake{} is PPAD-complete.
\end{theorem}
An earlier result by \citet{deng2012algorithmic} showed that \cake{} is PPAD-complete when preferences are not specified via utility functions but by means of more general `preferences sets', analogous to our model of the housing market.

We are now ready to reduce \cake{} to \rkkm{}, which implies that \rkkm{} and \housing{} are PPAD-hard. Our proof of \cref{thm:cake-to-rkkm} relies on the representation of cake preferences in terms of utilities.
\begin{theorem}
\label{thm:cake-to-rkkm}
There exists a polynomial-time reduction from $d$-\cake{} to $d$-\rkkm{}.
\end{theorem}
\begin{proof}
Suppose $(\varepsilon, u^1, \ldots, u^d)$ is an instance of $d$-\cake{}. We define a KKM covering $C^i = (C^i_1, \ldots, C^i_d$) of $\Delta_{d-1}$ for each player $i \in [d]$ by
\[
C^{i}_j \coloneqq \{\xb \in \Delta_{d-1} \mid u^i(I_{j}(\xb)) \geq u^i(I_k(\xb)) \text{ for all } k \in [d] \}.
\]
Hence, $C^{i}_j$ consists of all the cake cuts for which agent $i$ prefers piece $j$. To see that $C^i$ is indeed a KKM covering of $\Delta_{d-1}$, consider any $S \subseteq [d]$ and let $\xb \in F_S = \conv \{\bm{e}^k \mid k \in S \}$. By choice of $\xb$, we have $x_k = 0$ for all $k \notin S$ and $x_k > 0$ for some $k \in S$. The non-negativity property of the utility function $u^i$ then implies that player $i$ prefers some piece of the cake indexed by $S$ to all the pieces not indexed by $S$. It follows that $\xb \in C^i_k$ for some $k \in S$, so $C^i$ is a KKM covering.

We now construct an instance of \rkkm{}. Fix approximation parameter $\varepsilon' = \frac{\varepsilon}{4K}$ (where $K$ is the Lipschitz constant of the utility functions $u^i$), and let $g^1, \ldots, g^d$ be the KKM covering functions corresponding to $C^1, \ldots, C^d$. As $g^i(\xb,j) = 1$ if $\xb \in C^i_j$ and $g^i(\xb,j) = 0$ otherwise, it is straightforward to see that $g^i$ can be implemented efficiently with at most $d$ calls to the utility function $u^i$ of player $i$.

Suppose $(\xb, \pi)$ is a solution to this \rkkm{} instance $(\varepsilon', g^1, \ldots, g^d)$. Hence, there exists a collection of points $\xb^1, \ldots, \xb^d$ such that $\xb^i \in C^i_{\pi(i)}$ and $\| \xb - \xb^i \|_1 \leq \varepsilon'$ for all $i \in [d]$. We will now show that $(\xb,\pi)$ forms an $\varepsilon$-approximate solution to the \cake{} instance.

Fix a player $i \in [d]$. Our goal is to show that she $\varepsilon$-approximately prefers the $\pi(i)$-th piece of cut $\xb$ to all other pieces specified by $\xb$. Fix some $k \in [d]$, and the $k$-th pieces $I_k(\xb) = [a,b]$ and $I'_k(\xb) = [a',b']$ of cuts $\xb$ and $\xb^i$. Note that $a = \sum_{l=1}^{k-1} x_l$ and $a' = \sum_{l=1}^{k-1} x^i_l$, so $\|\xb - \xb^i \|_1 \leq \varepsilon'$ implies $|a-a'| \leq \varepsilon'$. Likewise, $|b-b'| < \varepsilon'$. By the Lipschitz continuity of the agent's utility function, we thus have
\[
u^i([a,b]) - u^i([a',b']) \leq K ( |a-a'| + |b-b'| ) \leq 2\varepsilon'K = \frac{\varepsilon}{2}.
\]
As player $i$ prefers the $\pi(i)$-th piece of cut $\xb^i$ by choice of $\xb^i$, it follows that $u^i(I_{\pi(i)}(\xb)) \geq u^i(I_k(\xb^i)) + \varepsilon$ for every $k \in [d]$.
\end{proof}

\citet{hollender2023envy} showed that \cake{} with a fixed number of players $d \geq 4$ has a query complexity lower bound of $\Omega(\poly(\frac{1}{\varepsilon}))$. The reduction in our proof of \cref{thm:cake-to-rkkm} thus implies the same lower bound on the query complexity for $d$-\housing{} and $d$-\rkkm{}.

\begin{corollary}
$d$-\housing{} and $d$-\rkkm{} with $d \geq 4$ have a query complexity of $\Omega(\poly(\frac{1}{\varepsilon}))$.
\end{corollary}

\subsection{The reduction from the \spernerX{} problem}
\label{section:sperner-to-rkkm}
We now develop the reduction from \spernerX{} to 3-\rkkm{}. We do this via the intermediary problem 3D-\kkm{}. As we know from \cref{section:equivalence-reductions} that $n$-\housing{} and $n$-\rkkm{} are computationally equivalent, \cref{thm:n-to-n+1} then implies that both $n$-\housing{} and $n$-\rkkm{} are PPAD-hard for any $n \geq 3$.

\paragraph{The \spernerX{} problem.}
Consider the triangulation of triangle $N \Delta_{2} = \conv \{N\eb^1, N\eb^2, N\eb^3 \}$ with side lengths $N$ into equilateral cells of side length $1$. This is illustrated for $N=4$ in \cref{fig:sperner-triangulation}. The vertices of this triangulation are $V_N \coloneqq \{\vb \in \N_0^3 \mid v_1 + v_2 + v_3 = N\}$. A \textit{coloring} $c:V_N \to \{1,2,3\}$ assigns each vertex one of three colors. It is a \textit{Sperner coloring} if $c(\vb) \neq i$ when $v_i=0$. That is, each corner $N\eb^i$ of the triangle is colored in $i$, and each point on the boundary between $N\eb^i$ and $N\eb^j$ is colored in $i$ or $j$. A cell of the triangulation is \textit{trichromatic} if each of its three vertices receives a different color. Sperner's lemma (\cref{lemma:sperners-lemma}) guarantees the existence of a trichromatic cell.

\begin{lemma}[\citep{sperner1928neuer}]
\label{lemma:sperners-lemma}
Every Sperner coloring of every triangulation of a triangle has at least one trichromatic cell.
\end{lemma}

This leads naturally to the computational problem of finding a trichromatic cell. We assume that the size $N$ of the triangulation is given as an input parameter in a binary representation, as the problem is easy otherwise.

\begin{problem}{\spernerX{}}
\textbf{Input:} The size $N \in \N$ of the triangle, and a Sperner coloring function $c:V_N \to \{1,2,3\}$ of the vertices of the triangulation of $N\Delta_2$.

\textbf{Output:} The three vertices of a trichromatic cell in the (unit) triangulation of $N \Delta_2$.
\end{problem}

The \spernerX{} problem stated above is computationally equivalent to the version considered in \citep{chen2009complexity}, which considers a triangulation of a right-angled triangle. It is straightforward to see that we can reduce between these problems. \citet{chen2009complexity} settled the computational complexity of \spernerX{}.

\begin{theorem}[\citep{chen2009complexity}]
    \spernerX{} is PPAD-complete.
\end{theorem}

We now turn to the problem of reducing from \spernerX{} to \kkm{}. Fix the triangulation of $N \Delta_2$ and a Sperner coloring $c$. For simplicity, consider the variant of \kkm{} that seeks to find an $\varepsilon$-approximate KKM solution given a KKM covering of $N\Delta_2$, and $N$ is given as an input parameter alongside the covering function. This is equivalent to our original definition of \kkm{} with approximation parameter $\varepsilon' = \frac{\varepsilon}{N}$. Define the following KKM covering of $N\Delta_2$. For $i \in \{1,2,3\}$, let
\begin{equation}
\label{eq:Sperner-KKM-covering}
C_i \coloneqq \{ \xb \in N\Delta_2 \mid c(\vb)=i \text{ for some } \vb \in \argmin_{\vb \in V} \| \xb-\vb \|_1\}.
\end{equation}
Hence, $C_i$ contains all the points whose closest vertex is $i$-colored. The three covering sets are illustrated in \cref{fig:sperner-triangulation}; each $C_i$ is the union of identically-colored hexagons. It is straightforward to check that $(C_1, C_2, C_3)$ is a KKM covering of $N\Delta_2$, i.e.~that $F_S = \conv \{N\eb^i \mid i \in S\}$ is contained in $\bigcup_{i \in S} C_i$ for every $S \subseteq [n]$. Indeed, fix $S \subseteq [n]$ and $\xb \in F_S$, so $x_k = 0$ for all $k \notin S$. If $S = \{1\}$, then $\xb$ is a vertex. As $c$ is a Sperner coloring, we have $c(\xb)=1$ and $\xb \in C_1$. If $S = \{1,2\}$, then $\xb$ lies on the boundary of $N\Delta_2$ between $N\eb^1$ and $N\eb^2$. The closest vertex $\vb \in V$ to $\xb$ also lies on this side and, as $c$ is a Sperner coloring, we have $c(\vb) \in \{1,2\}$. Hence, $\xb \in C_1 \cup C_2$. Finally, if $S = \{1,2,3\}$, then $\xb$ lies in the interior. By construction, $\xb \in C_{c(\vb)}$ for some vertex $\vb$ closest to $\xb$. So $\xb \in C_1 \cup C_2 \cup C_3$.

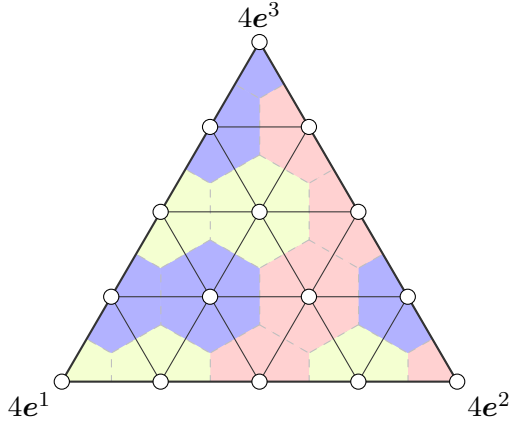
\begin{figure}
    \centering
    \input{figs/triangulation.tikz}
    \caption{Triangulation of simplex $\Delta_{2}$ with side length $4$ into regular triangles with side length $1$, drawn using solid black lines. The three corresponding covering sets $C_1, C_2$ and $C_3$ consist of the yellow, red and blue regions, respectively.}
    \label{fig:sperner-triangulation}
\end{figure}

We now establish, in \cref{lemma:cell-is-trichromatic}, that a solution $\xb$ to the instance of $\kkm{}$ with approximation parameter $\varepsilon=\frac{1}{8}$ and KKM covering $(C^1, C^2, C^3)$ as defined in \eqref{eq:Sperner-KKM-covering} lies in the interior of a trichromatic cell. This forms the main part of our reduction. For this result, we will need the intuitive result of \cref{lemma:closest-vertex} that the closest vertex to a point in a cell, or to a point `almost in a cell', is one of the vertices of this cell. Note that we can apply a distance-preserving map to $N \Delta_2$ that maps any given cell $\Delta$ to~$\Delta_2$. Every vertex $\vb$ in the resulting triangulation satisfies $\vb \in \Z^n$ and $v_1 + v_2 + v_3 = 1$. We make use of this transformation to simplify our proofs of \cref{lemma:cell-is-trichromatic,lemma:closest-vertex}.

\begin{lemma}
\label{lemma:cell-is-trichromatic}
Fix a point $\xb$ in cell $\Delta$ of the triangulation. If its neighborhood $N_{\sfrac{1}{8}}(\xb)$ intersects with $C_1, C_2$ and $C_3$, then $\xb$ lies in the interior of $\Delta$, and $\Delta$ is trichromatic.
\end{lemma}
\begin{lemma}
\label{lemma:closest-vertex}
If $\xb$ lies in $N_{\frac{1}{8}}(\yb)$ for some point $\yb$ of a cell $\Delta$, then one of the three vertices of this cell~$\Delta$ is a closest vertex to $\xb$.
\end{lemma}
\begin{proof}
Apply a distance-preserving map to $N \Delta_2$ so that $\Delta = \Delta_{2}$. Let $\vb \notin \{ \eb^1, \eb^2, \eb^3 \}$ be a vertex of the triangulation. Recall that its entries are integral and $v_1 + v_2 + v_3 = 1$. We now show that $\| \xb - \vb \|_1 \geq 2-\frac{1}{8}$ and $\| \xb - \eb^i\| \leq \frac{4}{3}$ for some $i \in \{1,2,3\}$.

Suppose first that $\xb$ lies on the boundary of $\Delta$. Without loss of generality, $\xb$ is a convex combination of $\eb^1$ and $\eb^2$ and $x_1 \geq x_2$, so $x_3 = 0$ and $x_1 + x_2 = 1$. Hence $\|\xb-\eb^1\|_1 = 1 - x_1 + x_2 \leq 1$, so the $L_1$-distance of $\xb$ to one of the vertices of $\Delta$ is at most $1$. Next, we show that $\|\xb-\vb\|_1 = |x_1 - v_1| + |x_2 - v_2| + |v_3| \geq 2$ by distinguishing between different possibilities for $\vb$. If $v_1, v_2 \geq 1$, then $v_3 \leq -1$ and so $\|\xb - \vb\|_1 \geq 2 - x_1 - x_2 + |v_3| \geq 2$. If $v_1, v_2 \leq 0$, with at least one strict inequality, then $v_3 \geq 2$ and so $\|\xb-\vb\|_1 = x_1 + x_2 - v_1 - v_2 + |v_3| \geq 4$. If $v_1 = 0$ and $v_2 \geq 2$, then $v_3 \leq -1$ and so $\|\xb-\vb\|_1 = x_1 +x_2 + v_2 + |v_3| \geq 2$. If $v_1 < 0$ and $v_2 = 1$, then $v_3 \geq 1$ implies $\|\xb-\vb\|_1 \geq 2$. Finally, if $v_1 < 0$ and $v_2 > 1$, then $v_2 - v_1 \geq 3$ guarantees the same.

Now suppose that $\xb$ lies in the interior of $\Delta$. Without loss of generality, let $x_1 \geq x_2 \geq x_3$, so $\|\xb - \eb^1 \|_1 \leq \|\xb - \eb^2 \|_1$ and $\|\xb - \eb^1 \|_1 \leq \|\xb - \eb^3 \|_1$. Moreover, $\|\xb - \eb^1 \|_1 = 1 - x_1 + x_2 + x_3$ is maximized when $\xb$ is the barycenter of $\Delta_2$, so $x_1 = x_2 = x_3 = \frac{1}{3}$. Hence $\|\xb - \eb^i\|_1 \leq \frac{4}{3}$. Let $\yb$ be the point at which the line segment from $\xb$ to $\vb$ intersects the boundary of $\Delta$. It is straightforward that $\|\xb - \vb \|_1 = \|\xb-\yb\|_1 + \|\yb-\vb \|_1 \geq \|\yb-\vb\|_1 \geq 2$, where the second inequality follows from the previous paragraph.

Finally, suppose $\xb$ does not lie in $\Delta$. By construction, there exists $\yb$ on the boundary of $\Delta$ with $\|\xb - \yb\|_1 \leq \frac{1}{8}$. By the triangle inequality and the previous paragraph, the distance from $\xb$ to $\eb^1$, $\eb^2$ or $\eb^3$ is at most $1+\frac{1}{8}$, and the distance from $\xb$ to $\yb$ is at least $2-\frac{1}{8}$, so the result holds.
\end{proof}

\begin{proof}[Proof of \cref{lemma:cell-is-trichromatic}]
Without loss of generality, let $\Delta$ be the unit simplex $\Delta_2$ (by applying a suitable distance-preserving map) and suppose each of its vertices $\eb^i$ is $i$-colored. Let $\zb \coloneqq (\frac{1}{3}, \frac{1}{3}, \frac{1}{3})$ be the barycenter of~$\Delta$.

Assume that $N_{\sfrac{1}{8}}(\xb) \cap C_i \neq \emptyset$, so there exist $\xb^i \in C_i$, for all $i \in \{1,2,3\}$. Suppose we can show that $\xb$ lies in $N_{\sfrac{1}{2}}(\zb)$. Then $\xb^i \in N_{\sfrac{5}{8}}(\zb)$ for every $i$. As $N_{\sfrac{5}{8}}(\zb)$ is contained in $\Delta$, the points $\xb, \xb^1, \xb^2$ and $\xb^3$ then all lie in the interior of~$\Delta$. But as $\xb^i \in C_i$, and $\eb^1, \eb^2$ and $\eb^3$ are the three closest vertices to any point in $\Delta$ by \cref{lemma:closest-vertex}, it then follows from the definition of $C^i$ that $\Delta$ is trichromatic.

It remains to show that $\xb \in N_{\sfrac{1}{2}}(\zb)$. Suppose, for the sake of contradiction, that $\xb \notin N_{\sfrac{1}{2}}(\zb)$. This means that $\|\xb - \zb \|_1 = |x_1 - \frac{1}{3}| + |x_2 - \frac{1}{3}| + |x_3 - \frac{1}{3}| > \frac{1}{2}$. Without loss of generality, assume that $x_1 \geq x_2 \geq x_3$. As $0 \leq \xb \neq \zb$ and $x_1 + x_2 + x_3 = 1$, we have $x_1 > x_3$. Below, we will argue that
\begin{equation}
    \label{eq:trichromatic-cell}
    x^3_1 > x^3_3 \text{ or } x^3_2 > x^3_3.
\end{equation}
This implies $\|\xb^3 - \eb^3\|_1 - \|\xb^3 - \eb^1\|_1 = x^3_1 - x^3_3 > 0$ or $\|\xb^3 - \eb^3\|_1 - \|\xb^3 - \eb^2\|_1 = x^3_2 - x^3_3 > 0$. The fact that $\xb^3$ is thus closer to $\eb^1$ or $\eb^2$ leads to the contradiction $\xb^3 \notin C^3$.

We now argue that \eqref{eq:trichromatic-cell} holds. As $\|\xb - \xb^3 \|_1 \leq \frac{1}{8}$, we have $|x_i - x^3_i| \leq \frac{1}{8}$ for every $i$. Hence \eqref{eq:trichromatic-cell} holds immediately if $x_1 - x_3 > \frac{1}{4}$ or $x_2 - x_3 > \frac{1}{4}$. It remains to prove this last claim. As $\xb$ lies in $\Delta$, we have $x_1 + x_2 + x_3 = 1$. Hence, $x_1 \geq x_2 \geq x_3$ and $x_1 > x_3$ implies $x_1 > \frac{1}{3}$ and $x_3 < \frac{1}{3}$. It follows that $\|\xb - \zb \|_1 = x_1 - x_3 + |x_2 - \frac{1}{3}| > \frac{1}{2}$, so $x_1 - x_3 > \frac{1}{4}$ or $|x_2 - \frac{1}{3}| > \frac{1}{4}$. If $x_1 - x_3 > \frac{1}{4}$, we are done. So suppose $x_1 - x_3 \leq \frac{1}{4}$ and $|x_2 - \frac{1}{3}| > \frac{1}{4}$. This implies $x_2> \frac{1}{4} + \frac{1}{3}$ or $x_2 < \frac{1}{3} - \frac{1}{4}$. Suppose the latter holds, so $x_2 < \frac{1}{3} - \frac{1}{4}$. We also bound $x_1$ from above by noting that $x_1 - x_3 \leq \frac{1}{4}$ and $x_3 < \frac{1}{3}$ imply $x_1 < \frac{1}{4} + \frac{1}{3}$. It follows that $x_1 + x_2 + x_3 < \frac{1}{4} + \frac{1}{3} + \frac{1}{3} - \frac{1}{4} + \frac{1}{3} = 1$, a contradiction. Hence, we must have $x_2 > \frac{1}{4} + \frac{1}{3}$. It follows that $x_2 - x_3 > x_2 - \frac{1}{3} > \frac{1}{4}$, and we are done.
\end{proof}

\begin{theorem}
\label{thm:sperner-to-kkm}
There exists a polynomial-time reduction from \spernerX{} to 3D-\textsc{KKM}.
\end{theorem}
\begin{proof}
Suppose $(N, c)$ is a \spernerX{} instance. We construct an instance $(\varepsilon, g)$ of \kkm{} with $\varepsilon=\frac{1}{8}$. The covering function $g$ is associated with the KKM covering $C_1, C_2, C_3$ of $N \Delta_{2}$ constructed from $u$ as described in \eqref{eq:Sperner-KKM-covering}. Note that for any point $\xb \in N\Delta_{2}$, we have $g(\xb,i) = 1$ if a nearest vertex to $\xb$ is $i$-colored, and $g(\xb,i) = 0$ otherwise. Given an arbitrary point $\xb \in N\Delta_2$, we can round up or down the entries of $\xb$ to determine its nearest vertices $v^1 = (\lfloor x_1 \rfloor, \lceil x_2 \rceil, \lceil x_3 \rceil)$, $v^1 = (\lceil x_1 \rceil, \lfloor x_2 \rfloor, \lceil x_3 \rceil)$ and $v^3 = (\lceil x_1 \rceil, \lceil x_2 \rceil, \lfloor x_3 \rfloor)$. (If $\xb$ is a vertex of the triangulation, $v_1, v_2$ and $v_3$ will coincide; if $\xb$ lies on a face separating two cells, two of these vertices will coincide.) Then, with at most $3$ calls to the coloring function, we can determine whether $v^1$, $v^2$ or $v^3$ are $i$-colored. This implements the covering function $g$ efficiently.

Now suppose $\xb$ is a solution to the \textsc{KKM} instance $(\varepsilon,g)$, so there exists a point $\xb^i \in C_i$ with $\| \xb^i - \xb\|_1 \leq \frac{1}{8}$ for every $i \in \{1,2,3\}$. The point $\xb$ lies in some cell $\Delta$ of the triangulation. By \cref{lemma:cell-is-trichromatic}, $\xb$ lies in the interior of $\Delta$, and the cell is trichromatic. As above, we can compute its vertices $v^1, v^2$ and $v^3$ efficiently by rounding up and down to recover a solution to the \spernerX{} instance $(N,c)$.
\end{proof}

\begin{theorem}
\label{thm:kkm-to-rkkm}
There exists a polynomial-time reduction from 3D-\kkm{} to 3-\rkkm{}.
\end{theorem}
\begin{proof}
Suppose $(\varepsilon, N, g)$ is an instance of KKM associated with KKM covering $C = (C_1, C_2, C_3)$ of $N\Delta_{2}$. We reduce to an instance $(\varepsilon', h^1, h^2, h^3)$ of $3$-\rkkm{} with approximation parameter $\varepsilon' \coloneqq \frac{\varepsilon}{N}$ and three identical copies of the KKM covering $(D_1, D_2, D_3)$ of $\Delta_2$ defined by $D_i = \frac{1}{N} C_i$. A ${\varepsilon'}$-approximate solution $\xb$ to the latter, together with any permutation $\pi$ of $[3]$, is an $\varepsilon$-approximate solution to the former, and the covering functions $h^i$ for the $3$-\rkkm{} instance can be constructed efficiently.
\end{proof}

We can now combine the equivalence of $n$-\housing{} and $n$-\rkkm{} shown in \cref{section:equivalence-reductions} with \cref{thm:n-to-n+1,thm:kkm-to-rkkm} to get the following hardness result.
\begin{corollary}
\label{corollary:sperner-to-rkkm}
$n$-\housing{} and $n$-\rkkm{} are PPAD-hard for any $n \geq 3$.
\end{corollary}

\section{Membership of \housing{} and \rkkm{} in PPAD}
\label{section:membership}
We now show that $\rkkm{}$ lies in PPAD by reducing it to \sperner{}. By \cref{thm:housing-to-rkkm}, this also immediately implies that \housing{} lies in PPAD. Our reduction uses the same labeling and coloring technique employed by \citet{deng2012algorithmic} to reduce \cake{} to \sperner{}.

Suppose $C^1, \ldots, C^n$ is a family of KKM coverings. Like \citet{deng2012algorithmic}, we first we define the \textit{Kuhn triangulation} of the `large' cube $[0,N]^{n-1}$ obtained by dividing the cube into $N^{n-1}$ unit cubes, and then subdividing each unit cube into $(n-1)!$ simplicial cells. Next, we label \textit{and} color each vertex of this triangulation with one of colors $[n-1]_0$ and labels $[n-1]_0$.%
\footnote{Unlike \citet{deng2012algorithmic}, we label and color the entire triangulated large cube, and not just a subdivision of it. This allows us to reduce our problem to the standard \sperner{} problem defined on the subdivision of a cube.}
In order to define our labeling rule and coloring function, we first divide the `large' cube into $(n-1)!$ `large' simplices. Then we associate each vertex $\vb$ with its barycentric coordinates $\alpha(\vb)$ w.r.t.~the vertices of the large simplex in which $\vb$ lies. The labeling rule and coloring function are then specified with respect to $\alpha(\vb)$. We use the same labeling rule as \citet{deng2012algorithmic}; it guarantees that the vertices of a cell all receive different labels. Our coloring function then assigns each vertex $\vb$, labeled $i$, the smallest color $j$ for which $\alpha(\vb)$ lies in covering set $C^i_j$. Sperner's lemma implies that there exists a panchromatic cell. Thus for every label $i$, the $i$-labeled vertex of this cell is colored differently. Finally, as the $L_1$-distance between the barycentric coordinates of any two vertices in a cell of the Kuhn triangulation is at most $\frac{n}{N}$, we arrive at an $\frac{n}{N}$-approximate Rainbow-KKM solution to the \rkkm{} problem with KKM coverings $C^1, \ldots, C^n$.

\subsection{Kuhn's triangulation}

\begin{figure}
    \centering
    \input{figs/cube.tikz}
    \caption{The Kuhn triangulation of the unit cube.}
    \label{fig:kuhn-triangulation}
\end{figure}
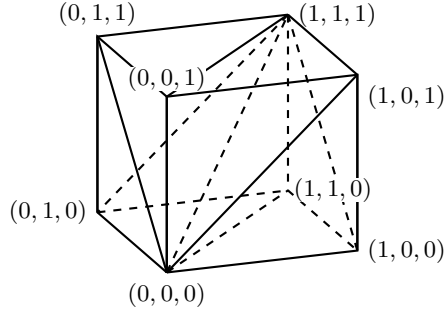

We describe Kuhn's method for triangulating a cube introduced by \citet{scarf1982computation} and \citet{deng2012algorithmic}. We begin by defining Kuhn's triangulation of the $n$-dimensional unit cube $[0,1]^n$. For each permutation $\pi$ of $[n]$, define the base simplex $\hat{\Delta}_{\pi} \coloneqq \{\xb \in [0,1]^n \mid x_{\pi(1)} \geq x_{\pi(2)} \geq \cdots \geq x_{\pi(n)}\}$. The vertices of this simplex are given by $\vb^{0}_{\pi} \coloneqq \bm{0}$ and $\vb^i_{\pi} \coloneqq \vb^{i-1}_{\pi} + \eb^{\pi(i)}$ for all $i \in [n]$. It is straightforward that the unit cube consists of the union of $\hat{\Delta}_{\pi}$ taken over all permutations $\pi$, and that the interiors of the base simplices do not overlap. The Kuhn triangulation of the unit cube is illustrated for $n=3$ in \cref{fig:kuhn-triangulation}.

Next, fix a positive integer $N \in \N$. We describe a way to triangulate the cube $[0,N]^n$ of side length $N$. We also call this the Kuhn triangulation of the ``large cube''. First, divide the cube into $N^d$ unit cubes. Then triangulate each unit cube into \emph{base simplices} (or \emph{cells}) by applying Kuhn's triangulation as described in the previous paragraph. Hence each unit cube contains $n!$ cells, one for each permutation $\pi$. In particular, if $\xb$ is the smallest point in the cube, then its cell associated with permutation $\pi$ is $\xb + \hat{\Delta}_{\pi}$, and is the convex hull of vertices $\vb^0_{\pi} \coloneqq \xb$ and $\vb^{i}_{\pi} \coloneqq \vb^{i-1}_{\pi} + \eb^{\pi(i)}$ for all $i \in [n]$. The set of all vertices of the triangulation is $V_N \coloneqq [n]_0^n$.

Alternatively, we can also divide the large cube $[0,N]^n$ into simplices $N \hat{\Delta}_{\pi} \coloneqq \{ N \xb \mid \xb \in \hat{\Delta}_{\pi} \}$ for all $\pi$. (This is analogous to performing Kuhn's triangulation of the unit cube, scaled up by a factor of $N$.) We will refer to the $n!$ resulting simplices $N \hat{\Delta}_{\pi}$ as the ``large simplices''. Kuhn's triangulation of the large cube has the important property, stated in \cref{lemma:cell-in-large-simplex}, that every cell is contained entirely in one of the large simplices. For more details on Kuhn's triangulation, see \citet{deng2012algorithmic} and \citet{scarf1982computation}.

\begin{lemma}[\citep{deng2012algorithmic}]
\label{lemma:cell-in-large-simplex}
Every cell of Kuhn's triangulation of the large cube $[0,N]^n$ is contained in a large simplex $N \hat{\Delta}_{\pi}$ for some permutation $\pi$ of $[n]$.
\end{lemma}

For every vertex $\vb \in V_N$, let $\alpha(\vb)$ denote its barycentric coordinates $\alpha(\vb) = (\alpha(\vb)_0, \ldots, \alpha(\vb)_n)$ w.r.t.~the vertices $v^{0}_{\pi}, \ldots, v^{n}_{\pi}$ of the large simplex $N \hat{\Delta}_{\pi}$ in which $\vb$ lies. Note that $\alpha(\vb)$ is $(n+1)$-dimensional and satisfies $\alpha(\vb) \geq \bm{0}$ and $\sum_{i \in [n]_0} \alpha(\vb)_i = 1$. If a vertex lies in two large simplices $N\hat{\Delta}_{\pi}$ and $N\hat{\Delta}_{\tau}$, then \cref{lemma:barycentric-coordinates-coincide} states that its barycentric coordinates w.r.t.~the vertices of $N\hat{\Delta}_{\pi}$ and $N\hat{\Delta}_{\tau}$ are identical. This ensures that $\alpha(\vb)$ is well-defined. Moreover, \cref{lemma:barycentric-coordinates-close} shows that the $L_1$-distance between the barycentric coordinates of two vertices of a cell is at most $\frac{n+1}{N}$. The proof of \cref{lemma:barycentric-coordinates-close} is essentially identical to the proof in \citet{deng2012algorithmic}, and is stated for completeness.

\begin{lemma}
\label{lemma:barycentric-coordinates-coincide}
Suppose $\xb$ lies in $N\hat{\Delta}_{\pi}$ and $N\hat{\Delta}_{\pi}$. Then
\[
\sum_{i=0}^n \alpha_i v^i_{\pi} = \sum_{i=0}^n \alpha_i v^i_{\tau}
\]
for some barycentric coordinate vector $\alpha \in \Delta_n$.
\end{lemma}
\begin{proof}
We will construct a sequence $\pi = \pi^0, \pi^1, \ldots, \pi^n = \tau$ of permutations of $[n]$ that satisfies three properties. Let $\bm{\alpha}^k$ be the barycentric coordinates of $\xb$ w.r.t.~the vertices $v^0 = \bm{0}$ and $v^{i}_{\pi^k} = v^{i-1}_{\pi^k} + \eb^{\pi(i)}$ ($\forall i \in [n]$) of $N \hat{\Delta}_{\pi^k}$. The three properties are:
\begin{enumerate}[(i)]
    \item \label{prop:one} $\pi^{k}_1, \ldots, \pi^k_k = \tau_1, \ldots, \tau_k$ for all $k \in [n]$,
    \item \label{prop:two} $\xb \in N\hat{\Delta}_{\pi^k}$ for all $k \in [n]_0$,
    \item \label{prop:three} $\bm{\alpha}^k = \bm{\alpha}^{k-1}$ for all $k \in [n]$.
\end{enumerate}
These three properties immediately imply the lemma.

We now define the sequence of permutations. Let $\pi^0 \coloneqq \pi$. We then iteratively define $\pi^k$ from $\pi^{k-1}$ as follows. Intuitively, the $k$-th iteration moves the element $\tau(k)$ in permutation $\pi^{k-1}$ to the $k$-th position of $\pi^k$. Formally, we let $j \in [n]$ be the index for which~$\pi^{k}(j) = \tau(k)$ and define
\[
    \pi^k(l) \coloneqq
    \begin{cases}
        \pi^{k-1}(j) & \text{if } l = k, \\
        \pi^{k-1}(l-1) & \text{if } l \in \{k+1, \ldots, j\},\\
        \pi^{k-1}(l) & \text{else}\\
    \end{cases}
\]
It is immediate that property \eqref{prop:one} holds, and that $j \geq k$ in every iteration; so the sequence of permutations is well-defined. We now prove property \eqref{prop:two} by induction on $k$. The base case $\xb \in \hat{\Delta}_{\pi^0}$ holds by definition. Now suppose that $\xb \in N\hat{\Delta}_{\pi^{k-1}}$, and let $j$ be as chosen above to define $\pi^k$. Hence, $x_{\pi^{k-1}(k)} \geq \cdots \geq x_{\pi^{k-1}(j)}$. Moreover, $x \in N\hat{\Delta}_{\tau}$ implies $x_{\pi^{k-1}(j)} \geq x_{\pi^{k-1}(k)}$ by our choice of $j$. It follows that $x_{\pi^{k-1}(k)} = \cdots = x_{\pi^{k-1}(j)}$, so $\xb \in N\hat{\Delta}_{\pi^k}$. Finally, we prove property \eqref{prop:three}. For this, note that the vertices corresponding to $\pi^{k-1}$ and $\pi^{k}$ satisfy $v^{i}_{\pi^{k-1}} = v^{i}_{\pi^{k}}$ for $i \in \{1, \ldots, k-1\}$ and $i \in \{j, \ldots, n\}$ by construction of the vertices of $N\hat{\Delta}_{\pi^{k-1}}$ and $N\hat{\Delta}_{\pi^{k}}$. Moreover, $\xb \in N\hat{\Delta}_{\pi^{k-1}}$ and $\xb \in N\hat{\Delta}_{\pi^{k}}$ inductively imply $\alpha^{k-1}_i = \alpha^k_i$ for $i \in \{j, \ldots, n\}$, then $\alpha^{k-1}_i = 0 = \alpha^k_i$ for $i \in \{k, \ldots, i-1\}$, and finally $\alpha^{k-1}_i = \alpha^k_i$ for $i \in \{1, \ldots, k-1\}$.
\end{proof}

\begin{lemma}[cf.,~\citet{deng2012algorithmic}]
\label{lemma:barycentric-coordinates-close}
For any two vertices $\yb$ and $\zb$ of a cell, $\|\alpha(\yb) - \alpha(\zb)\|_1 \leq \frac{n+1}{N}$.
\end{lemma}
\begin{proof}
Let $\yb$ and $\zb$ be two vertices of a cell $\Delta$ that is contained in large simplex $N\hat{\Delta}_{\pi}$. Without loss of generality, we assume that $\pi = (1, \ldots, n)$. (For other permutations, the result holds by symmetry). The vertices of the large simplex are $\vb^0 = \bm{0}$ and $\vb^i = \vb^{i-1} + \eb^i$ for every $i \in [n]$.

Let $\xb$ be the smallest point in the unit cube containing $\Delta$, and the $n+1$ vertices of $\Delta$ be $\wb^0_{\tau} = \xb$ and $\wb^i_{\tau} = \wb^{i-1}_{\tau} + \eb^{\tau(i)}$ ($\forall i \in [n]$) for some permutation $\tau$. Without loss of generality, we let $\yb = \wb^k$ and $\zb = \wb^l$ with $k < l$. By construction of the vertices of the large simplex, the barycentric coordinates of $\wb^i_{\tau}$ for $i \in [n]$ are given by $\alpha(\wb^{i}_{\tau}) = \alpha(\wb^{i-1}_{\tau}) + \frac{1}{N} \eb^{\tau(i)} - \frac{1}{N} \eb^{\tau(i)-1}$. Hence,
\[
\alpha(\zb) - \alpha(\yb) = \frac{1}{N} \sum_{i=k+1}^l \left (\eb^{\tau(i)} - \eb^{\tau(i)-1} \right) \in \frac{1}{N} \{ -1, 0, 1 \}^{n+1}.
\]
\end{proof}

Starting from vertex $\bm{0}$ and using the same argument as in the proof of \cref{lemma:barycentric-coordinates-close}, we can iteratively show for every vertex $\vb \in V_N$ that the entries of $\alpha(\vb)$ can be expressed as a proper fraction with denominator $N$.
\begin{corollary}
\label{corollary:barycentric-coordinates}
For any vertex $\vb \in V_N$ and $i \in [n]_0$, we have $\alpha(\vb)_i \in \left \{ \frac{1}{N}, \frac{2}{N}, \ldots, 1 \right \}$.
\end{corollary}

Finally, we use the mapping $\alpha$ of vertices to their barycentric coordinates to define a labeling of the vertices. We use the same labeling rule as \citet{deng2012algorithmic}, but extend it to the entire triangulated large cube. \Cref{lemma:Simmons-Su-labeling} shows that, under this labeling, every cell's vertices are labeled with all numbers $[n]_0$. (We restate the proof from \citet{deng2012algorithmic} for completeness.) Such a labeling is called a \textit{Simmons-Su labeling}. Note also that $\L(\vb)$ is well-defined, as $N\alpha(\vb)_i$ is integral by \cref{corollary:barycentric-coordinates}. For each $\vb \in V$, let 
\begin{equation}
\label{eq:simmons-su-labeling}
\L(\vb) \coloneqq \sum_{i \in [n]_0} i N \alpha(\vb)_i \mod n+1.
\end{equation}

\begin{lemma}[cf.,~\citet{deng2012algorithmic}]
\label{lemma:Simmons-Su-labeling}
The labeling $\L$ is a Simmons-Su labeling.
\end{lemma}
\begin{proof}
Suppose $\Delta$ is a cell in the triangulation of the large cube. As in the proof of \cref{lemma:barycentric-coordinates-close}, we see that the barycentric coordinates of the vertices $\vb^0_{\pi}, \ldots, \vb^n_{\pi}$ of $\Delta$ satisfy $\alpha(\vb^i_\pi) = \alpha(\vb^{i-1}_\pi) + \frac{1}{N}\eb^{\pi(i)} - \frac{1}{N}\eb^{\pi(i)-1}$. Hence every vertex of $\Delta$ receives a different label, as
\[
\L(\vb^i_\pi) \equiv \L(\vb^{i-1}_{\pi}) + \pi(i) - (\pi(i)-1) \equiv \L(\vb^{i-1}_{\pi})+1 \mod n+1.
\]

\end{proof}

\subsection{The \textsc{Sperner} problem}
\label{section:sperner-problem}
We now define the \sperner{} problem for arbitrary dimensions. In contrast to \spernerX{}, which was defined using a triangulated triangle, \citet{papadimitriou1994complexity} introduces the \sperner{} problem for arbitrary dimensions using a `large' cube divided into unit cubes. Suppose, for instance, that we are given a (Kuhn) triangulation of the large cube $[0,N]^n$. A coloring $c:V_N \to [n]_0$ of its vertices is a \textit{Sperner coloring} if it satisfies
\begin{enumerate}
    \item $c(\vb) \neq i$ if $v_i = 0$, for all $i \in [n]$,
    \item $c(\vb) \neq 0$ if $v_i = N$ for any $i \in [n]$.
\end{enumerate}

By embedding the triangulated cube into the triangulation of an $n$-dimensional simplex, it can be shown that Sperner's lemma implies the existence of a \textit{panchromatic} cell. That is, a cell whose $n+1$ vertices all receive a different color. This inspires the following total search problem, analogous to \spernerX{}. \citet[Proof of Theorem 3]{papadimitriou1994complexity} shows that \sperner{} lies in PPAD.

\begin{problem}{\sperner{}}
\textbf{Input:} Dimension $n \in \N$ and side length $N \in \N$ specifying Kuhn triangulation of large cube $[0,N]^n$. Sperner coloring function $c : [N]_0^n \to [n]_0$.

\textbf{Output:} The vertices of a panchromatic cell.
\end{problem}

\begin{theorem}[\citep{papadimitriou1994complexity}]
\label{thm:sperner-in-PPAD}
\sperner{} is in PPAD.
\end{theorem}

\subsection{The reduction from \rkkm{} to \sperner{}}
Suppose $C^1, \ldots, C^n$ is a family of $n$ KKM coverings of $\Delta_n$ and let $\varepsilon > 0$ be an approximation parameter. For notational convenience, index the sets of each covering $C^i$ by $[n-1]_0$ instead of $[n]$, so $C^i = (C^i_0, C^i_1, \ldots, C^i_{n-1})$. Without loss of generality, we assume that the KKM coverings are sparse (cf.~\cref{section:rkkm}). Consider the Kuhn triangulation of the cube $[0,N]^{n-1}$ with $N = \frac{n}{\varepsilon}$ and let its vertices be the set $V$. Note that the cube is $(n-1)$-dimensional. The vertices of this triangulation are labeled with labels $[n-1]_0$ according to labeling $\L$. The key to our reduction is to define the coloring $c:V \to [n-1]_0$ as
\begin{equation}
\label{eq:sperner-coloring}
    c(\vb) = \min \left \{j \in [n-1]_0 \mid \alpha(\vb) \in C^{\L(\vb)}_j \right \}.
\end{equation}

We now argue that $c$ is a Sperner coloring by verifying both conditions (1) and (2) of the definition in \cref{section:sperner-problem}. Fix some vertex $\vb \in V$. If $v_j = 0$ for some $j \in [n-1]$, then $\alpha_j = 0$, and so $\alpha \in F_{[n-1]_0 \setminus \{j\}}$. As the KKM coverings are sparse, we have $\alpha \notin C^i_j$ for all coverings $C^i$. This verifies condition (1). Now suppose $v_j = N$ for some $j \in [n-1]$. This implies $\sum_{k=1}^{n-1} \alpha_k = 1$, and so $\alpha_0 = 0$. By the same argument as above, we get $\alpha \notin C^i_0$ for all coverings $C^i$. This verifies condition (2).

Hence there exists a cell of the triangulation that is panchromatic, so every vertex of the cell has a different color. Fix such a cell. As $\L$ is a Simmons-Su labeling, each vertex of the cell also has a different label. Label these vertices $\vb^0, \ldots, \vb^{n-1}$ so that $\L(\vb^i) = i$. Define the permutation $\pi$ of $[n-1]_0$ as $\pi(i) = c(\vb^i)$. Then $\alpha(\vb^i) \in C^i_{\pi(i)}$. By \cref{lemma:barycentric-coordinates-close}, the $L_1$-distance between any two $\vb^i$ and $\vb^j$ is at most $\frac{n}{N} = \varepsilon$. It thus follows that any $\vb^i$ together with $\pi$ form an $\varepsilon$-approximate Rainbow-KKM solution for the family of KKM coverings $C^1, \ldots, C^n$.

\begin{theorem}
\label{thm:rkkm-to-sperner}
There is a polynomial-time reduction from \rkkm{} to \sperner{}.
\end{theorem}
\begin{proof}
Suppose $(\varepsilon, g^1, \ldots, g^n)$ is an instance of $\rkkm{}$, and each $g^i$ is associated with KKM covering $C^i$. It remains to argue that we can efficiently construct a routine that implements coloring function $c$ using functions $g^i$. Recall the definition of the coloring $c$ from \eqref{eq:sperner-coloring}. In order to determine the color of vertex $\vb$, we first compute its barycentric coordinates $\alpha$. Then compute its label $l = \sum_{i \in [n-1]_0} i N \alpha(\vb) \mod n$. Finally, determine the smallest $j \in [n-1]_0$ for which $\alpha \in C^l_j$, using at most $n$ calls to the covering function $g^l$.
\end{proof}

In the white-box model, \cref{thm:sperner-in-PPAD,thm:rkkm-to-sperner} imply that \housing{} and \rkkm{} lie in PPAD. Moreover, the reduction also immediately implies a brute-force approach for \rkkm{} (and thus also for \housing{}): compute the color $c(\vb)$ of each vertex $\vb \in V$ and check for each cell whether it is panchromatic. In the black-box model, this algorithm has an exponential query complexity $O(N^n)$ in the input size $\log \frac{1}{\varepsilon}$ of the approximation parameter $\varepsilon$. We can improve this upper bound on the query complexity of \housing{} and \rkkm{}. \Citet{deng2011discrete} show that \sperner{} has a query complexity of $\Theta(N^{n-1})$. Our reduction in the proof of \cref{thm:rkkm-to-sperner} solves an instance of \rkkm{} with $O(N^{n-1})$ queries to the coloring function $c$ and every such query costs at most $n$ queries to the covering function.

\begin{corollary}
\housing{} and \rkkm{} have query complexity $O(N^{n-1})$.
\end{corollary}

\section{Conclusion}\label{section:conclusion}
We have shown that \housing{} and the closely related \rkkm{} are PPAD-complete. Whereas the housing market is tractable with two agents, we show that PPAD-completeness holds even when we restrict the market to $n \geq 3$ agents with identical preferences. Moreover, we describe exponential upper and lower bounds for the query complexity of solving the two problems in the white-box model that hold when $n \geq 4$.

While our results may be perceived as negative for market design, it is not uncommon for PPAD-complete problems to be implemented for allocation in practice (e.g., approximate competitive equilibrium with equal incomes for course allocation). Rather we hope that our results will stimulate further examination of the computational complexity of allocation systems with income effects. Three sets of open questions naturally arise.

First, what is the computational complexity of \housing{} when preferences are provided as utility functions instead of preference sets? Does hardness continue to hold under the natural assumptions of monotonicity in money and no externalities as in \citet{quinzii1984core} and \citet{svensson1984competitive}? What is the complexity of computing competitive equilibrium in model with multi-good demand under assumptions, such as net substitutability condition, that guarantee existence \citep{baldwin2023equilibrium}?
Secondly, what is the exact mapping between the assumptions in \housing{} and the monotonicity assumptions on valuations in \cake{} imposed by \citet{deng2012algorithmic} and \citet{hollender2023envy} to achieve polynomial-time algorithms for cutting a cake with $3$ and $4$ players?
Thirdly, what algorithms might work well in practice for approximating equilibria in the presence of income effects?

\bibliographystyle{abbrvnat}
\bibliography{bib}
\end{document}

%% file: figs/triangulation.tikz
\begin{tikzpicture}
\pgfmathtruncatemacro{\N}{4};
\pgfmathsetmacro{\scaling}{1.3};
\tikzstyle{cset}=[
    regular polygon,
    draw=gray!50,
    dashed,
    regular polygon sides=6,
    minimum size=\scaling*1.16cm,
    inner sep=0,
    shape border rotate=30
]
\tikzstyle{triangulation}=[draw=black!80]
\tikzstyle{vertex}=[circle,draw=black,fill=white, inner sep=0, minimum size=0.2cm]

\coordinate (a) at (0,0);
\coordinate (b) at (\scaling*\N,0);
\coordinate (c) at ({\scaling*\N/2},{sqrt(3) * \scaling * \N / 2});
\def\bc(#1:#2:#3){(barycentric cs:a=#1,b=#2,c=#3)}

\draw [triangulation, thick] (a) -- (b) -- (c) -- cycle;
\node [below left]  at (a)    {$\N \eb^1$};
\node [below right] at (b)    {$\N \eb^2$};
\node [above=2]       at (c)    {$\N \eb^3$};

\foreach \x in {1,...,\N} { \draw[triangulation] \bc(\N-\x:\x:0) -- \bc(\N-\x:0:\x); }
\foreach \x in {1,...,\N} { \draw[triangulation] \bc(\N-\x:\x:0) -- \bc(0:\x:\N-\x); }
\foreach \x in {1,...,\N} { \draw[triangulation] \bc(\N-\x:0:\x) -- \bc(0:\N-\x:\x); }

\foreach \x in {0,...,\N} {
    \pgfmathtruncatemacro{\a}{\N-\x};
    \foreach \y in {0,...,\a} {
        \pgfmathtruncatemacro{\z}{\N-\x-\y};
        \node[vertex] at \bc(\x:\y:\z) {};
    }
}
\begin{scope}[on background layer]
    \clip (a) -- (b) -- (c);
    \node[cset, fill=blue!30] at \bc(0:0:4) {};
    \node[cset, fill=red!18] at \bc(0:1:3) {};
    \node[cset, fill=red!18] at \bc(0:2:2) {};
    \node[cset, fill=blue!30] at \bc(0:3:1) {};
    \node[cset, fill=red!18] at \bc(0:4:0) {};
    \node[cset, fill=blue!30] at \bc(1:0:3) {};
    \node[cset, fill=lime!18] at \bc(1:1:2) {};
    \node[cset, fill=red!18] at \bc(1:2:1) {};
    \node[cset, fill=lime!18] at \bc(1:3:0) {};
    \node[cset, fill=lime!18] at \bc(2:0:2) {};
    \node[cset, fill=blue!30] at \bc(2:1:1) {};
    \node[cset, fill=red!18] at \bc(2:2:0) {};
    \node[cset, fill=blue!30] at \bc(3:0:1) {};
    \node[cset, fill=lime!18] at \bc(3:1:0) {};
    \node[cset, fill=lime!18] at \bc(4:0:0) {};
\end{scope}
\end{tikzpicture}

%% file: figs/cube.tikz
\begin{tikzpicture}[font=\footnotesize]
\tikzstyle{vx}=[black, fill=white, inner sep=0]
\tikzstyle{wire}=[black, thick]
\tikzstyle{bg}=[black, thick, dashed]
\begin{axis}[
    width=5cm,
    height=5cm,
    view={-20}{20},
    xmin=0, xmax=1,
    ymin=0, ymax=1,
    zmin=0, zmax=1,
    hide axis,
    clip=false,
    ]

    \draw[wire] (0,0,0) -- (0,0,1) -- (1,0,1) -- (1,0,0) -- (0,0,0);  
    \draw[wire] (0,0,0) -- (0,0,1) -- (0,1,1) -- (0,1,0) -- (0,0,0);  
    \draw[wire] (1,0,0) -- (1,0,0) -- (1,0,1) -- (1,1,1);  
    \draw[wire,bg] (1,1,1) -- (1,1,0) -- (1,0,0);  
    \draw[wire] (0,1,0) -- (0,1,1) -- (1,1,1);  
    \draw[wire,bg] (0,1,0) -- (1,1,0);  

    \draw[wire] (0,0,0) -- (0,1,1);
    \draw[wire] (0,0,0) -- (1,0,1);
    \draw[wire, dashed] (0,0,0) -- (1,1,0);
    \draw[wire, dashed] (0,0,0) -- (1,1,1);
    \draw[wire] (0,0,1) -- (1,1,1);
    \draw[wire, dashed] (1,0,0) -- (1,1,1);
    \draw[wire, dashed] (1,1,1) -- (0,1,0);

    \node[below] at (0,0,0) {$(0,0,0)$};
    \node[vx, above=2] at (0,0,1) {$(0,0,1)$};
    \node[left] at (0,1,0) {$(0,1,0)$};
    \node[right] at (1,0,0) {$(1,0,0)$};
    \node[above] at (0,1,1) {$(0,1,1)$};
    \node[below right] at (1,0,1) {$(1,0,1)$};
    \node[vx, right=2] at (1,1,0) {$(1,1,0)$};
    \node[right] at (1,1,1) {$(1,1,1)$};

\end{axis}
\end{tikzpicture}